\documentclass[a4paper]{article}
\usepackage[utf8]{inputenc}
\usepackage{amsfonts,amsmath,amssymb,amsthm,bm,bbm,mathtools,color,multirow,graphicx}
\usepackage[hidelinks]{hyperref}

\newcommand{\Prob}{\mathbb{P}}
\newcommand{\Ex}{\mathbb{E}}
\newcommand{\Var}{\mathrm{Var}}
\newcommand{\Cov}{\mathrm{Cov}}
\newcommand{\Mult}{\mathit{Mult}}

\DeclareMathOperator{\diag}{diag}
\newcommand{\bmp}{\bm{p}}
\newcommand{\bmu}{\bm{u}}

\newcommand{\bmx}{\bm{x}}
\newcommand{\bmX}{\bm{X}}
\newcommand{\bbmX}{\bar{\bm{X}}}
\newcommand{\ett}{\mathbbm{1}}
\newcommand{\calS}{\mathcal{S}}
\newcommand{\calR}{\mathcal{R}}
\newcommand{\calT}{\mathcal{T}}

\newtheorem{theorem}{Theorem}
\newtheorem{lemma}{Lemma}
\newtheorem{remark}{Remark}
\newtheorem{proposition}{Proposition}
\theoremstyle{definition}

\title{How to ask sensitive multiple choice questions}
\author{Andreas Lager{\aa}s\thanks{AFA Insurance, Sweden} \thanks{Department of Mathematics, Stockholm university, SE 106 91 Stockholm, Sweden} \and Mathias Lindholm\footnotemark[2] \thanks{Corresponding author: \href{mailto:lindholm@math.su.se}{lindholm@math.su.se}}}

\begin{document}

\maketitle

\begin{abstract}
Motivated by recent failures of polling to estimate populist party support, we propose and analyse two methods for asking sensitive multiple choice questions where the respondent retains some privacy and therefore might answer more truthfully. The first method consists of asking for the true choice along with a choice picked at random. The other method presents a list of choices and asks whether the preferred one is on the list or not. Different respondents are shown different lists. The methods are easy to explain, which makes it likely that the respondent understands how her privacy is protected and may thus entice her to participate in the survey and answer truthfully. The methods are also easy to implement and scale up.
\end{abstract}

\section{Introduction}

When asking someone about a personal deed or preference one would expect her to be less truthful, or willing to provide an answer, the more she thinks that that deed or preference is illegal or shameful. This could be one of the reasons many opinion polls have underestimated the public support for populist parties and candidates in several countries in recent years. The concrete example that motivates us is the larger than expected support for the Sweden Democrats (SD), a nationalist party, in the Swedish general election 2014.

This phenomenon could work in two ways to skew the result of a poll. Firstly, respondents may choose not to participate at all. This compounds the larger problem of reaching a representative sample of the population. Secondly, the respondents' answers may be biased towards less sensitive preferences. In this paper we suggest methods that let the respondents keep some privacy by introducing noise or asking for a less exact answer. Ideally this would affect both the participation rate and the bias, but we only quantitatively analyse the possible reduction in bias.

The issue of reaching a representative sample of respondents, due to both demographic reasons and the sensitivity of the questions asked, has been raised recently in the Swedish press \cite{lindholm2017dnDebatt}. For in-depth analysis of UK elections, see \cite{sturgis2017assessment,sturgis2016report} that comment on possible problems with sensitive answers (``shy Tories'') and also highlight the importance of using proper random samples.

Our particular interest is in genuine multiple choice questions, i.e.\ where the number of choices is larger than two. This is necessary in the context of a parliamentary system with proportional representation such as the Swedish one. In elections with a first-past-the-post system there are usually only two dominating parties. 

In the situation where there are only two choices of interest, say in an effectively two-party system, or when one wants to estimate a single proportion in a population, e.g.\ the proportion that has committed a certain crime, there are several methods to provide anonymity, such as the Randomised Response Technique (RRT) introduced by \cite{warner1965randomized}.

In multiple choice situations this technique is not directly applicable, but there are multiple choice extensions such as e.g.\ \cite{abul1967multi,eriksson1973new}, as well as other techniques for scrambling data reported by the persons contributing to the poll which can be found in e.g.\ \cite{chaudhuri2013indirect} and the references therein. These methods can be perceived as ``weird'' or hard to explain to the respondent due to the need of a complicated randomisation device, that the respondents true answer is neglected, or that unrelated questions are asked, see e.g.\ \cite{kuha2014item} and the references therein. This also raises the concern that these types of methods will in practice be difficult to implement on a sufficiently large scale, which is needed in the context of assessing voting intentions in nationwide general elections.

This motivates the need for methods to handle polling of sensitive multiple choice questions. The methods we describe focus on simple practical implementation. In particular the methods allows for (i) some degree of anonymity for the respondent, (ii) simple or no randomisation device, (iii) simple questions and (iv) the possibility of using automated surveys. Point (iv) is important, since this implies that it is inexpensive to scale up the size of the survey.

We will introduce and analyse two methods which we call the ``pair method'' and the ``list method''. In the pair method, the respondent is asked to name her party preference and another party chosen uniformly at random, and to reply with both parties in random order. The randomisation is done privately by the respondent. This method is similar to the methods of \cite{esponda2009surveys, esponda2016statistical}. For the list method, the respondent is presented with a list of several parties and asked whether her preferred party is on the list or not. Different respondents are presented with different lists. As opposed to the pair method, the list method falls into the category of ``non-randomized'' response techniques, see e.g.\ \cite{tan2009sample,yu2008two}.

The pair method provides anonymity since a respondent with the sensitive preference will also reply with a non-sensitive one, so that from the point of view of the interviewer she could also have had the non-sensitive preference as her true preference and only responded with the sensitive one by chance. The respondent has some plausible deniability. The list method provides even more anonymity than the pair method if the presented lists contain more than two parties. It is also easier to implement since it requires no randomisation by the respondent.

The flip side of anonymity is that less information is gained from each respondent and thus a larger sample is needed. The pair method is more efficient than the list method in this regard. We analyse the methods both from the perspective of the level of anonymity provided and the efficiency lost.

In the above argumentation we have only referred to ``anonymity'' in colloquial terms without defining this in more detail. Recall that for the simple RRT introduced in \cite{warner1965randomized} there is only a single sensitive answer in a dichotomous response situation. In \cite{leysieffer1976respondent} it was suggested that one could measure the degree of anonymity in the RRT setting using a measure called ``jeopardy'', which relates to how much information is revealed concerning the sensitive answer. In the present paper we discuss this measure in relation to the information theoretical concept of entropy, see e.g.\ \cite{kendall1973entropy,kullback1959statistics,mackay2003information}. In particular we discuss the problem with measures similar to jeopardy in the situation where, in the extreme case, all answers may be regarded as sensitive. An example is that people may be reluctant to reveal their true voting intention regardless of which political party they will place their vote on.

The objective of the methods is to provide anonymity in order to reduce bias. We therefore focus on unbiased estimators, which in our case turn out to be maximum likelihood estimators. Even though we do not perform a Bayesian analysis, it is worth noting that the measures of anonymity we use require some reference to the distribution of preferences in the population, so that if one wants to reason about the respondents' perceived anonymity one must consider their a priori views on this distribution. On the other hand, if you are reluctant to specifying, or not interested in, the perceived level of privacy, no opinion concerning the true a priori distribution of votes is needed.

Note that both suggested methods are very simple to explain to a respondent. In the language of \cite{kuha2014item} we believe that the suggested methods provide low levels of ``weirdness''. Moreover, as opposed to standard RRTs the respondent will {\it always} provide her true voting intention -- this is important, not the least w.r.t.\ the problem of getting respondents to participate at all, see e.g.\ \cite{lindholm2017dnDebatt,sturgis2017assessment,sturgis2016report}. That is, if a pollster manages to contact a reluctant potential respondent we believe that it is crucial that the respondent understands, at least intuitively, what is meant with ``anonymity'' and that the respondent's true intention is accounted for. From a practical perspective we believe that the pair method probably is more suitable to use either in face-to-face situations, since the pollster may provide a suitable randomization device, e.g.\ a box with cards where all parties are represented, or in a web based survey where the randomisation can be done in the respondent's browser. The list method, due to its yes/no character even could be implemented using cell phone text messages (SMS) -- ``{\it Would you consider voting for any of the political parties in the list provided below? Please reply to this text message with `Yes' or `No'.}''. Moreover, the list method is based on random sampling, and we believe it is inexpensive to implement using automated surveys. This is an important feature w.r.t.\ the concerns raised in \cite{sturgis2017assessment,sturgis2016report} on non-randomized sampling.

\section{Measures of anonymity and information}\label{sec:anon}

\subsection{Entropy, information, and privacy}\label{sec:entropy}

We will introduce the concepts entropy and information. For more on these topics see e.g.\ \cite{mackay2003information}. Let $T$ be a discrete random variable with probability function $p_T(t) \coloneqq \Prob[T=t]$, $t\in\mathcal{T}$, and define its \emph{entropy} by
$$
H[T] \coloneqq -\Ex_T[\log_2 p_T(T)] = -\sum_{t \in \calT} p_T(t) \log_2 p_T(t).
$$
The entropy measures the uncertainty about the outcome of $T$ in the sense that it gives bounds for the average number of yes/no-questions that are necessary to ascertain the outcome. The unit of measurement is called \emph{bit}. More specifically, with $Q(T)$ being the necessary number of dichotomous questions needed to ascertain the outcome,
$$
H[T] \leq \Ex[Q(T)] < H[T]+1.
$$
In our context, we think of $T$ as the true voting intentions of a randomly chosen respondent, i.e.\ the distribution of $T$ has support on a set of political parties. If the interviewer could only ask yes/no-questions of the kind ``does your preferred party belong to the set $\mathcal{S}$'' for different sets of parties $\mathcal{S}$, then the expected number of questions she would need would lie between $H[T]$ and $H[T]+1$. From the definition of $H[T]$ it is also clear that its maximum is attained when $p_T(t) = 1/|\calT|$ for all $t \in \calT$.

The respondent is afforded some degree of anonymity or privacy if she does not have to divulge all information about her intentions, but rather retain some bits of entropy. Let $R$ be another discrete random variable having joint probability function $p_{T,R}(t,r)$ with $T$. We will think of $R$ as the respondent's answer to the interviewer.

With $p_{T\mid R}(t\mid r)$ being the conditional probability function, we can define the \emph{joint entropy}
$$
H[T,R] \coloneqq -\Ex_{T,R}[\log_2 p_{T,R}(T,R)],
$$
the \emph{conditional entropy}
$$
H[T\mid R] \coloneqq -\Ex_{T,R}[\log_2 p_{T\mid R}(T\mid R)] = H[T,R] - H[R] 
$$
and the mutual information
$$
I[T;R] \coloneqq H[T] + H[R] - H[T,R] = H[T] - H[T\mid R].
$$

The entropy $H[T]$ then measures the interviewer's uncertainty about the voting intentions of a respondent \emph{before} responding and $H[T\mid R]$ the uncertainty \emph{after} having responded. The mutual information $I[T;R]$ measures how much the uncertainty has decreased due to receiving an answer. In other words, the conditional entropy $H[T\mid R]$ measures the amount of retained privacy and the mutual information $I[T;R]$ measures the amount of divulged information.

One can also note that $I[T;R]$ may be re-written according to
\begin{align*}
I[T;R] &= \sum_{t \in \calT, r \in \calR} p_{T,R}(t,r)\log_2\frac{p_{T,R}(t,r)}{p_T(t)p_R(r)}\\
&= D_{KL}(p_{T,R}~ ||~ p_Tp_R),
\end{align*}
where $D_{KL}(F~ ||~ G)$ corresponds to the {\it Kullback-Leibler divergence} between the probability distributions $F$ and $G$, see e.g.\ \cite[p.\ 34]{mackay2003information}. Further, it is clear that $D_{KL}(F~ ||~ G) \ge 0$, with equality iff $F = G$. Thus, $I[T; R] \ge 0$ and we only have equality iff $T$ and $R$ are statistically independent, i.e.\ by knowing $R$ no information is gained w.r.t.\ $T$ and vice versa. In our setting this corresponds to complete anonymity and will never be possible to attain for the methods below.

In the case $R=T$, when the respondent tells the interviewer her precise voting intentions, $H[T,R] = H[T]$ so that $H[T\mid R] = 0$ and there is no residual uncertainty or privacy. Likewise $I[T;R] = H[T]$, meaning all information has been divulged.

It must be noted that the measures of entropy and information are population averages. The individual respondent might be more interested in $-\log_2 p_T(t)$, which measures how uncommon her intention $t$ is and how much information about herself she would give away by revealing that. Likewise, she might only want to participate in the survey if $-\log_2 p_{T\mid R}(t\mid r)$, measuring her retained privacy, is high for all possible answers $r$ that she might be prompted by the survey design to give to the interviewer. This should be kept in mind when designing the survey.

The description above also made the tacit assumption that the distribution of $T$ is common knowledge. If that were the case, the survey wouldn't be needed in the first place! In order to obtain unbiased answers it is important that the respondents' perceived privacy is protected to some extent, and that means that one must consider the respondents' subjective distributions of $T$, and possibly their beliefs about the interviewer's belief etc. There is no way to quantify these subjective probabilities so we proceed pragmatically and assume that there is a rough agreement in the population about the distribution of $T$.

The two proposed methods described in this paper are easily analysed within this framework. Note that a single yes/no-question divulges at most one bit since the entropy of a two-point distribution is less than or equal to one, with equality in the case of equidistribution, as for a fair coin toss. Indeed, in the pairs method, when each respondent provides her true voting intention together with a randomly chosen other party, it only takes a single additional yes/no-question to ascertain her true intention, viz.\ ``is your true preference the first of the two parties in the pair?'' Therefore the retained privacy is at most one bit.

The list method is on the other side of the spectrum since it only asks a single yes/no-question, and therefore the amount of divulged information is at most one bit. If the lists are chosen to have support of close to half of the population the divulged information is close to one bit.

\subsection{Jeopardy}\label{sec:jeopardy}

Above the concept of mutual information was introduced as a measure of how much information about $T$ is revealed by answering $R$. As will become clear later on, we are mainly interested in situations where $R$ may be seen as a function of $T$. That is, given that a respondent's voting preference is $t$ its answer will follow the distribution $p_{R \mid  T}(r \mid t)$ for all $r \in \calR$, where $\calR$ is the set of all possible responses. Thus, if we let $\calS$ be the set of sensitive or stigmatizing preferences, it is clear that by applying Bayes' theorem we get
\[
\frac{p_{T \mid R}(\calS \mid r)}{p_{T \mid R}(\calS^c \mid r)} = \frac{p_{R \mid T}(r \mid \calS)}{p_{R \mid  T}(r \mid \calS^c)}\frac{p_T(\calS)}{p_T(\calS^c)},
\]
where
\[
p_T(\mathcal A) := \sum_{t \in \mathcal A}p_T(t), \quad \mathcal A \subset \mathcal{T},
\]
and
\[
p_{T | R}(\mathcal A \mid r) := \sum_{t \in \mathcal A}p_{T | R}(t | r), \quad \mathcal A \subset \mathcal{T},
\]
which rephrased in terms of information by using $\log_2$ yields
\begin{align}\label{eq : discriminating information}
\log_2\frac{p_{R \mid T}(r \mid \calS)}{p_{R \mid  T}(r \mid \calS^c)} = \log_2\frac{p_{T \mid R}(\calS \mid r)}{p_{T \mid R}(\calS^c \mid r)} - \log_2\frac{p_T(\calS)}{p_T(\calS^c)},
\end{align}
see e.g.\ \cite[Eq.\ (2.3)]{kullback1959statistics}. The ratio on the left-hand side of \eqref{eq : discriminating information} is what is called jeopardy, which was introduced in \cite{leysieffer1976respondent}:
\[
J(r) := \frac{p_{R \mid T}(r \mid \calS)}{p_{R \mid  T}(r \mid \calS^c)}, \quad \Prob(T \in \calS \cup \calS^c) = 1.
\]
$J(r)$ measures how much the unconditional odds have changed by answering $r$, or, in other words, how much the respondent is jeopardized by answering $r$. $J(r)$ is therefore called the \emph{jeopardy} with respect to $\calS$ \cite{leysieffer1976respondent}. In \cite{leysieffer1976respondent} one motivation for $J(r)$ in the case with only a dichotomous sensitive question is that $J(r)$ is independent of the true population proportions $p_i$, i.e.\ $J(r)$ only depends on the design probabilities of how responses are distributed given the respondent's position to the (single) sensitive answer. It is, however, important to note that in the situation when there are more than one sensitive alternative, $J(r)$ {\it will} depend on the true underlying population proportions as well. Hence, this situation applies to polling in political systems with more than two political parties where respondents want to keep their vote secret.

Another interpretation of \eqref{eq : discriminating information} is in terms of likelihood ratios, since $T \in \calS$ or $T \in \calS^c$ may be seen as two different parametrisations of a probability law, see e.g.\ \cite[pp.\ 4--5]{kullback1959statistics}. Hence, \eqref{eq : discriminating information} tells us how much information is revealed in favour of $T \in \calS$ opposed to $T \in \calS^c$ when providing the response $R = r$.

Further, as noted above, $J(r)$ is a measure of jeopardy for a single response. In order to overcome this limitation it has been proposed in \cite{chaudhuri2009protection} to measure the overall jeopardy of a survey design by averaging over the set $\mathcal{R}$ of all possible responses, and to consider the quantity
$$
\bar J \coloneqq \frac{1}{| \mathcal{R} |}\sum_{r\in\mathcal{R}}J(r).
$$

\begin{remark}

Recall the information theoretic interpretation of $I[T;R]$ in terms of Kullback-Leibler divergence and consider $\bar J$ and $J(r)$. A natural information theoretic extension of $J(r)$ for the situation with multiple sensitive answers is to consider
\[
\Ex_{R \mid T \in \calS}[\log_2 J(R)] = D_{KL}(p_{R \mid T \in \calS}~ ||~ p_{R \mid T \in \calS^c}) \ge 0,
\]
which corresponds to the mean information in favour of $T \in \calS$ over $T \in \calS^c$ when assuming that $T \in \calS$. For more on this, see e.g.\ \cite[Eq.\ (2.5), p.\ 5]{kullback1959statistics} and the surrounding discussion.
\end{remark}

\subsection{Variance, bias and bias-detection in practice}\label{sec:variance_bias}

In the case of elections a benchmark estimator $\tilde p_i$ for each party $i$ is the one obtained from fitting a binomial distribution, i.e.\ $\tilde p_i$ is the fraction of respondents who say they will vote for party $i$:
\[
n\tilde p_i \sim \mathit{Bin}(n, p_i - b_i),
\]
where $b_i$ is the bias ($\sum_ib_i = 0$).

Let $\hat p_i$ be an unbiased estimator of $p_i$, based on some anonymisation procedure. One would expect that $\Var[\hat p_i] > \Var[\tilde p_i]$ when both estimators are derived from samples of equal size, since the anonymity decreases the precision. However, this might be a price worth paying if anonomity produces a large enough reduction.

%
%

We will now describe how we may use both $\hat p_i$ and $\tilde p_i$ to detect bias in practice. Let us for a moment assume that we may calculate $\Var[\hat p_i]$ as well as obtain a Gaussian approximation of $\hat p_i$, we can calculate the power of the test of 
\[
\left\{\begin{array}{l}
\text{$H_0$: $b_i = 0$}\\
\text{$H_1$: $b_i > 0$}
\end{array}\right.
\]
based on the plug-in statistic
\begin{align}
T_i(b_i) := \frac{\hat p_i - \tilde p_i - b_i}{\sqrt{\Var_{\hat p_i}[\hat p_i] + \Var_{\tilde p_i}[\tilde p_i]}}  \sim \mathrm{asym.}~ N(0,1)\label{eq : power-statistic},
\end{align}
given that the true bias is $b_i$ and where $\Var_p[\cdot]$ is calculated assuming the true proportion is given by $p$. Note that the statistic $T_i(b_i)$ from \eqref{eq : power-statistic} is only truly computable in practice given that $b_i = 0$. The asymptotic power of this test is hence given by
\begin{align}
\pi_\gamma(b_i) :=&~ \Prob(\mathrm{reject}~ H_0 \mid H_1)\nonumber\\
=&~ \Prob(T_i(0) > z_\gamma \mid H_1: b_i > 0) \nonumber\\
\stackrel{\mathrm{asym.}}{=}&~ 1 - \Phi(z_{1-\gamma} - \frac{b_i}{\sqrt{\Var_{\hat p_i}[\hat p_i] + \Var_{\tilde p_i}[\tilde p_i]}}).\label{eq : power-bias}
\end{align}
Thus, by using \eqref{eq : power-bias} we can e.g.\ assess the size of $n$ needed in order to obtain a specific level of power to detect a certain level of bias.

An important remark concerning \eqref{eq : power-bias} is that  $\pi_\gamma(b_i)$ is a decreasing function in $\sqrt{\Var_{\hat p_i}[\hat p_i] + \Var_{\tilde p_i}[\tilde p_i]}$, which in itself typically is a decreasing function in terms of $p_i$. This, hence, implies that it ought to be easier to detect bias for parties with few intended voters.

\section{The multinomial distribution}\label{sec:mult}

Since it is easier to work with a multinomial distribution rather than a multivariate hypergeometric distribution, we will assume that we poll with replacement or that the population size is infinite. This is an innocuous assumption for the applications we have in mind. We recall some well-known properties of the multinomial distribution.

Let $\bmX = (X_1,\dots,X_N)'\sim \Mult(n,\bmp)$ with $\bmp = (p_1,\dots,p_N)'$ and $\sum_k p_k = 1$ so that
\begin{align*}
\Prob(\bmX = \bmx) &= \binom{n}{x_1,\dots,x_N}p_1^{x_1}\cdots p_N^{x_N} \\
\Ex[\bmX] &= n\bmp \\
\Var[\bmX] &= n(\diag(\bmp) -\bmp\bmp')
\end{align*}
To ease the notation we introduce $V(\bmp) \coloneqq \diag(\bmp) -\bmp\bmp'$. The maximum likelihood estimator of $\bmp$ is 
\[
\tilde\bmp:=\frac1n\bmX,
\]
with
\begin{align}
\Ex[\tilde\bmp] &= \bmp \notag\\
\Var[\tilde\bmp] &= \frac1n V(\bmp) \label{ml_var}
\end{align}
Let $\mathcal{A}(M,N)$ be the set of $M\times N$ matrices with non-negative elements and column sums all equal to 1. 

The different protocols we describe all have in common that we want to draw inference about a vector of probabilities $\bmp=(p_1,\dots,p_N)'$ when the data comes from a multinomial distribution $\Mult(n,\bmu)$ with $\bmu=A\bmp$ for a known matrix $A$, or more generally, when the data comes from independent multinomials $\Mult(n_i,A_i\bmp)$ for $i=1,\dots,L$. We need $A_i\in\mathcal{A}(M_i,N)$ so that $\bmu_i$ is a vector of probabilities.
\begin{lemma}\label{lemmaUB}
Let $\bm{X}_i = (X_{i1},\dots,X_{iK_i})\sim \Mult(n_i,\bm{u}_i)$ independently for $i=1,\dots,L$, with $\bm{u}_i = A_i\bm{p}$ for given $A_i\in\mathcal{A}(K_i,N)$ and $\alpha_i := n_i / n$ where $n = \sum_i n_i$. Denote by 
$$
A \coloneqq \begin{pmatrix} \alpha_1 A_1 \\ \vdots \\ \alpha_L A_L \end{pmatrix}\quad\text{and}\quad \bbmX \coloneqq  \begin{pmatrix} \bmX_1 \\ \vdots \\ \bmX_L \end{pmatrix}.
$$
If $A$ has rank $N$, then
\begin{equation}\label{phat_def}
\hat{\bm{p}} \coloneqq \frac{1}{n}(A'A)^{-1}A'\bbmX
\end{equation}
is an unbiased estimator of $\bmp$ with variance
\begin{align}
\Var[\hat{\bmp}] &= \frac{1}{n}(A'A)^{-1}\bigg(\sum_{i=1}^L \alpha_i^3 A_i'V(A_i\bmp)A_i\bigg)(A'A)^{-1}. \label{phat_var2}
\end{align}
In the case $L=1$ with $A:=A_1$ and $n\coloneqq n_1$, $\hat\bmp=\frac1n A^+\bmX_1$ and
$$
\Var[\hat{\bmp}] = \frac1n\big((A'A)^{-1}A'\diag(A\bmp)A(A'A)^{-1} -\bmp\bmp'\big).
$$
\end{lemma}
See Section \ref{sec:proofs} for proof of this Lemma.

Note that the definition of $\hat \bmp$ from \eqref{phat_def} is very natural w.r.t.\ $\bmu = A\bmp$, since $\hat \bmp$ is merely the standard least squares regression coefficient estimator of $\bmp$. Moreover, due to this interpretation it is reasonable that the variance of $\hat \bmp$ will also average out deviations between observations relating to individual $p_i$'s --- manifested by the multiplication from left and right by $(A'A)^{-1}$, which is a global quantity affecting all components of $\bmp$. That is, given that we do not obtain direct observations of individual $p_i$'s, which in our situation is a consequence of anonymisation, poor precision relating to one component of $\bmp$ will to some degree contaminate the remaining estimators' precision as well. Also note that $\hat \bmp$ is a linear transformation of $\bbmX$ and hence there is no guarantee that $\hat p_i$ is non-negative for all $i$. Still, it is important to note that since $\bbmX/n$ is a maximum likelihood estimator of $A\bmp$ it follows that the above $\hat \bmp$ is an ML estimator of $\bmp$.

Further, given Lemma \ref{lemmaUB} it is reasonable to expect that there exists a corresponding central limit theorem, which there is:

\begin{lemma}\label{lemma: CGS}
Let $\hat{\bm{p}}$ be defined according to Lemma \ref{lemmaUB}. Assume that
\[
\alpha_i = \frac{n_{i}}{n}\to \tilde\alpha_i> 0, \quad \text{as } n_i,n\to\infty,
\]
for all $i$ such that $\sum_i \tilde\alpha_i = 1$ and denote
\[
\tilde A \coloneqq \begin{pmatrix} \tilde\alpha_1 A_1 \\ \vdots \\ \tilde\alpha_L A_L \end{pmatrix}.
\]
Then
\begin{align*}
\sqrt{n}(\hat{\bm{p}}  - \bm{p}) \stackrel{D}{\to} \mathbf{G},\quad \text{as } n,n_i\to\infty, \quad \text{for all } i,
\end{align*}
where $\mathbf{G}$ is multivariate Gaussian with mean $\mathbf{0}$ and covariance
\[
\Var[\mathbf{G}] = (\tilde{A}'\tilde{A})^{-1}\bigg(\sum_{i=1}^L \tilde{\alpha}_i^3 A_i'V(A_i\bmp)A_i\bigg)(\tilde{A}'\tilde{A})^{-1}.
\]
\end{lemma}
The proof of Lemma \ref{lemma: CGS} is based on a standard central limit theorem for the multinomial $\bmX_i$s, which combined with Slutsky's theorem and a general version of the continuous mapping theorem which provides convergence for sequences of mappings yields the desired result. A detailed proof is given in Section \ref{sec:proofs}.

\section{Introducing two new anonymised survey methods with open answers: The Pair method and the List method}

\subsection{Pair method}

Consider the following polling protocol: each individual participating in the poll is asked to name the party she intends to vote for together with an additional party chosen uniformly at random amongst the remaining $N-1$ parties. The answers are reported unordered, i.e.\ the interviewer does not know which party is the true vote intention.

Note that we here assume that there is only a single sensitive party to vote for. That is, it should not be possible to obtain a voting pair $\{i,j\}$ where both parties are regarded as being sensitive. We will return to this situation when we discuss the list method.

The possible answers are the $M := \binom{N}{2} = \frac{N(N-1)}{2}$ unordered pairs $\{1,2\}$, $\{1,3\}$, \dots, $\{N-1,N\} = P_1,\dots,P_M$. When we need to order all pairs we always use this lexicographic ordering, i.e.\ $\{i,j\}$, where $i<j$. Let $b_{ik}=1$ if $i\in P_k$ and $b_{ik} = 0$ otherwise.

Assuming that the true voting intentions have the frequencies $p_1,\dots,p_N$, the probability of receiving the answer $\{i,j\}$ is given by 
\begin{align*}
u_{ij} \coloneqq p_i\tfrac{1}{N-1} + \tfrac{1}{N-1}p_j = \tfrac{1}{N-1}(p_i+p_j)
\end{align*}
since either $i$ or $j$ must be the true intention with respective probability $p_i$ and $p_j$, and in either case the random choice, $j$ and $i$ respectively, has probability $\frac{1}{N-1}$. If the poll size is $n$ and $X_{ij}$ is the number of answers $\{i,j\}$, then clearly $\bmX = (X_{12},X_{13},\dots,X_{N-1,N})'\sim \Mult(n,\bmu)$ with $\bmu := (u_{12},u_{13},\dots,u_{N-1,N})'$. The ML estimator for $\bmu$ is $\hat\bmu\coloneqq \frac1n\bmX$. This can be used to derive an unbiased estimator of $\bmp$, since we can re-write $\bmu = A\bmp$, where $A$  is the $N\times M$ matrix  defined according to:
$$
A \coloneqq \frac{1}{N-1}B' = \frac{1}{N-1}\begin{pmatrix} 1 & 1 & 0 & \dots & 0 & 0 \\
                                         1 & 0 & 1 & \dots & 0 & 0\\
                                         \vdots & \vdots & \vdots & & \vdots & \vdots \\
                                         0 & 0 & 0 & \dots & 1 & 1 \end{pmatrix},
$$
where $B_{ik} = b_{ik}$ with $b_{ik}$ given as above. That is, by using the above definition of $A$ it is clear that Lemma \ref{lemmaUB} applies to the pair method, which yields the following explicit result:

\begin{theorem}\label{thm_pair}
The estimator
\begin{equation}\label{phat_par}
\hat{p}_i \coloneqq \frac{N-1}{N-2}\sum_j \hat u_{ij} - \frac{1}{N-2}
\end{equation}
is an unbiased estimator of $p_i$, $i = 1,\ldots,N$, with
\begin{align}
\Var[\hat p_i] &= \frac1n\left(\frac{1+(N-3)p_i}{(N-2)} -p_i^2\right)\label{phat_par_var}\\
\Cov[\hat p_i, \hat p_j ] &= -\frac1n\left(\frac{1-p_i-p_j}{(N-2)^2}+p_ip_j\right),\quad i\neq j.\label{phat_par_cov}
\end{align}
\end{theorem}
A detailed proof of Theorem \ref{thm_pair} is given in Section \ref{sec:proofs}.

\subsection{List method}\label{sec:list}

If there are two sensitive choices, the pair method will for some respondents produce an answer with both those choices. To reintroduce some deniability, the method could be extended to triples so that each true preference is accompanied by two random choices. If there are three sensitive choices, one could ask for quadruples, and so on. This quickly becomes unwieldy.

Instead of asking the respondent to produce a list of several parties, where one is the true preference and all other are random, the interviewer might simply show a list of parties and ask if the respondent's preference is on the list or not. If the interviewer picks lists from a well-constructed set of lists, it is possible to derive an estimate of the population preferences. We will shortly describe what we mean with ``well-constructed''. Note that we assume a fixed set of $N$ possible choices so that a no-answer means that the preferred party must be on the complementary list of parties.

If the respondent answers truthfully, the probability of a yes-answer is $u_i \coloneqq \sum_{k\in \mathfrak{L}_i^+}p_k$ when she is presented with a list $\mathfrak{L}_i^+$. We can put this in the notation of Section \ref{sec:mult} if we let $A_i$ be a $2\times N$ matrix with elements $a_{1k} = \ett\{k\in\mathfrak{L}_i^+\}$ and $a_{2k}=\ett\{k\notin\mathfrak{L}_i^+\}$, so that $\bmu_i = (u_i, 1-u_i)'=A_i\bmp$ and the number of (yes, no) answers from asking $n_i$ people about the list $\mathfrak{L}_i^+$ is $\Mult(n_i,\bmu_i)$. Note that the first row of $A_i$ codes for membership in the list $\mathfrak{L}_i^+$ and the second codes for membership in the complementary list $\mathfrak{L}_i^- \coloneqq \{i:i\notin \mathfrak{L}_i^+\} = \{1,\dots,N\}\setminus \mathfrak{L}_i^+$.

In order to apply Lemma \ref{lemmaUB} we need the matrix $A$ that stacks all $L$ matrices $A_1$ to $A_L$ for the lists $\mathfrak{L}_1^+$ to $\mathfrak{L}_L^+$ to have rank $N$, i.e.\ the space spanned by its $N$ column must have full dimension $N$. This is what we mean with a ``well-constructed'' set of lists.

In order to see how this can be done in practice, let us consider the case with $N = 4$, the first situation which provides anonymity when there is only a single sensitive alternative, by constructing lists with two parties in each. That is, we have $\binom{N}{2} = 6$ such combinations of $\mathfrak{L}_i^+$ and $\mathfrak{L}_i^-$, i.e.\ 3 lists in total, which stacked gives us
$$
A \coloneqq  \begin{pmatrix} \alpha_1 A_1\\
\alpha_2 A_2\\
\alpha_3 A_3 \end{pmatrix} = \begin{pmatrix} \alpha_1 & \alpha_1 & 0  & 0 \\
                                         0 & 0 & \alpha_1 & \alpha_1\\
                                         \alpha_2 & 0 & \alpha_2 & 0 \\
                                         0 & \alpha_2 & 0 & \alpha_2\\
                                         \alpha_3 & 0 & 0 & \alpha_3\\
                                         0 & \alpha_3 & \alpha_3 & 0\end{pmatrix}, \quad \text{where } \alpha_i>0\text{ and }\sum_i\alpha_i=1.
$$
$A$ has full column rank (= 4). Thus, we can write $\bmu = A\bmp$ and Lemma \ref{lemmaUB} applies, which directly gives us that $\hat \bmp$ given by \eqref{phat_def} is an unbiased estimator of $\bmp$ together with computational formulas for the estimator's covariance.

Consequently, in the general situation, if we have a well-constructed set, the estimator $\hat\bmp$ defined in Equation \eqref{phat_def} is unbiased and has variance given by \eqref{phat_var2}. We can compare this variance with that of the ordinary ML-estimate of $\bmp$ for some different sets of lists. As $N$ grows the number of well-constructed sets grows exponentially. 

One idea is to construct lists with a priori voting support close to 50 \% (maximize anonymity) or with, as close to, equally many parties on all lists (low weirdness). Another criterion used to choose lists could be to minimise the variance. Note that all these possibilities amounts to choosing the $\alpha_i$'s in a certain way. Here one can also note that the $\alpha_i$'s allow for ex post calibration of potential non-response. In order to avoid too many subjective choices we will from now on primarily focus on the situation where we use all $\binom{N}{\lfloor N/2\rfloor}$ possible lists with (close to) equally many parties on each list and its complementary list.

An interesting result with an even number of parties is the following:

\begin{proposition}\label{prop:lika_varians}
Let $N=2M$ be even and the set of lists include all $L=\binom{N-1}{M-1}=\binom{N}{M}/2$ lists with $M$ parties that include party one, so that the set of complementary lists also cover $M$ parties each, and all exclude party one. The lists and complementary lists then cover all $\binom{N}{M} = \binom{N}{N/2}$ combinations of half of the $N$ parties. Further, assume that $\alpha_l = \frac1L$ for all $l$. Then all estimators $\hat p_i$ of the list method have the same variance.
\end{proposition}

The value of the common variance depend on the true $\bmp$, but it is still remarkable that one can get a common margin of error for all parties.

\subsection{On relations to other methods}

Let us consider the pair method when there is only a single sensitive answer. Assume that our aim is to maximize the plausible deniability of this sensitive voting intention, belonging to party 1, say. Given any answer $r_{1j}$ we could consider replacing the uniform probability with which your true voting intention is paired with an anonymising scrambling vote. Let $\pi_{ij}$ denote the probability that you add $j$ as a possible voting preference to your answer given that your true voting intention is $i$. A partial analog of a standard Forced Response Technique, see \cite{boruch1972relations}, is obtained by setting $\pi_{i1} \equiv 1, i\neq 1$, and let $\pi_{1j}\in [0,1], \sum_j\pi_{1j} = 1$. That is, given that your true voting intention is $1$ you add another possible voting preference $j$ according to $\pi_{1j}$, but if your true voting intention is $i\neq 1$ you will always, i.e.\ you are ``forced'' to,  add the sensitive voting alternative $1$ to your response. Thus, given this mechanism for producing response pairs, you may only give an answer in the set $\{1,j\}, j=2,\ldots,N$. It is, however, also clear from construction that this is only a reasonable method given that there is exactly one sensitive voting preference. Moreover, the choice of $\pi_{ij}$ still remains. Given no a priori knowledge of $p_i$ a natural choice of $\pi_{1j}$ is still $\pi_{1j} = 1/(N-1)$, see also the discussion on entropy in Section \ref{sec:entropy} and \ref{sec : entropy numeric}. 

One may also note that the pair method coincides with the ``Equiprobable Design Matrix'' method in \cite{esponda2016statistical}, although their definition of the method is in some sense complementary. That is, the pair method in their setting corresponds to choosing $N-2$ parties uniformly at random amongst those not containing the true voting intention. Hence, the remaining set which is not chosen contains two parties where one is the true voting intention together with an additional vote. On the other hand, this additional vote can due to exchangeability be seen as chosen uniformly at random amongst the $N-1$ parties not containing the true vote.

The same method from \cite{esponda2016statistical} covers answering with $n$-tuples with $n-1$ random choices added to the true preference.

\section{Evaluation of the Pair method and the List method: entropy, anonymity, bias and real-world examples}

We will now apply the theory from Section \ref{sec:anon} on the pair method and the list method. Focus will be on results and conclusions which can be drawn from these -- all derivations are given in the appendix. We will, when relevant, provide analytical expressions, but also numerical expressions for the situation when the true voting intention is uniform, i.e.\ $p_i = 1/N$ for all $i$, as well as $p_i$ corresponding to the general election in Sweden in 2014. The reason for analysing the situation with uniform voting preferences is that this situation will have maximal entropy, i.e.\ $H[T] = \log_2 N$, and hence have maximal anonymity. Moreover, in the Swedish general election of 2014 the voting preference of the Sweden Democrats (SD) turned out to be the hardest to estimate correctly. Due to this our numerical illustrations for measures of privacy will be w.r.t.\ to SD. The outcome of the Swedish general election from 2014 is given in Table \ref{tab: SWE election}. For the numerical illustrations we use equally balanced lists, i.e.\ a design with all lists that have 5 choices out of the possible 10. Thus, the number of lists and complementary lists is $\binom{10}{5} = 252$.

\subsection{Entropy}\label{sec : entropy numeric}

All entropy calculation for the pair and the list method are given in Table \ref{tab: theoretical entropy} for easy reference. To start off, recall that $I[T;R]$ is the expected reduction in uncertainty of $T$ due to knowledge of $R$. Due to construction, we expect that the pair method will disclose more information than the list method. By inspecting the results in the situation with equal voting preferences it is seen that the pair method only retains a single bit of information given an answer whereas the list method only reveals one bit. This is reasonable, since given that a respondent answers truthfully it is clear that the true voting intention is contained in the pair with which she has answered. Thus, a single yes/no question remains in order to fully disclose the respondents true voting preference. Concerning the list method it is on the other hand clear that very little information will be revealed when providing the yes/no answer to a specific list -- corresponding to revealing a single bit. This argumentation also explains the remaining measures listed in Table \ref{tab: theoretical entropy} in the case with uniformly distributed voting intentions.

If we instead turn to Table \ref{tab: SWE entropy}, where the true voting intentions are not uniformly distributed, it is seen that for the pair method very little information about the true voting intention is retained given an answer. This is perhaps not surprising given the discussion above, but it is interesting to see that the least anonymity contained when the true voting intention is SD is approximately 0.1 bits, when SD is paired with the answer ``O'' corresponding to ``Other parties''. Concerning the list method, we see that the situation is similar, but still more information is retained as a consequence of that {\it all} lists contain 5 out of 10 parties.

\subsection{Jeopardy}

From Table \ref{tab: theoretical jeopardy} it is clear that both $J(r)$ and $\bar J$ are greater than or equal to one for both methods. This is easily seen for the pair method, assuming $1 \in r$, since
\[
J(r) = \frac{1-p_1}{p_j} = \frac{p_j + \sum_{k \not\in \{1,j\}} p_k}{p_j} \ge 1,
\]
and the same argument applies to the other expressions as well. Thus, from the definition of jeopardy this tells us that an answer will increase the odds of the respondent to actually have the sensitive voting preference. Moreover, from Table \ref{tab: theoretical jeopardy} we see that for the situation with uniform voting preferences the pair method will reveal more information about the true voting intention than the list method -- in agreement with the results on entropy.

If we instead consider the situation corresponding to the Swedish general election from 2014, i.e.\ Table \ref{tab: SWE jeopardy}, we see that the same ordering of the methods w.r.t.\ $\bar J$ remains. Again, the sensitive vote is taken to be the Sweden Democrats (SD). One can also note that the pair method performs much worse than the list method in this situation. It is, however, less clear how to value a deviation from 1 as opposed to the entropy measures which have a natural scale in terms of bits.

Recall from Section \ref{sec:jeopardy} that the original introduction of jeopardy from \cite{leysieffer1976respondent} in the yes/no situation was possible to express in terms of ``design'' probabilities -- subjectively {\it chosen} probabilities which defines the randomization procedure being used. From Table \ref{tab: theoretical jeopardy} we see that this clearly is not the case for the pair and the list method, since both $J(r)$ and $\bar J$ depend on the true underlying voting intentions. Given this we believe that the entropy measures introduced above, which also depend on the true voting preferences, are closer to what a respondent is truly interested in. Due to this, we recommend that the entropy measures should be used instead of jeopardy in the present situation. 

\subsection{Variance and bias}

Assuming no bias, $b_i = 0$, then $\tilde p_i \coloneqq X_i/n$ is the standard unbiased ML estimator of $p_i$ assuming a direct response survey in an infinite population, i.e.\ $X_i \sim \mathit{Bin}(n, p_i)$, with $\Var[\tilde p_i] = \frac1n p_i(1 - p_i)$. Note that the variance of $\hat p_i$ for the pair method from \eqref{phat_par_var} above may be re-written according to
\begin{align*}
\Var[\hat p_i] &= \frac{1}{n}p_i(1-p_i) + \frac{1}{n(N-2)}(1-p_i)\\
&= \Var[\tilde p_i]  + \frac{1}{n(N-2)}(1-p_i)\\
&> \Var[\tilde p_i],
\end{align*}
as anticipated, since $\hat p_i$ is based on anonymised information which results in lack of precision compared to $\tilde p_i$.

Although the list method has computable (co)variances, it is in general not possible to obtain closed form expressions for a particular voting intention unless additional simplifying assumptions are made. In Table \ref{tab: theoretical variance bias} this is done for the case with uniform voting intentions. Compared with the results on entropy and jeopardy above, we see that the pair method here will outperform the list method. This is not surprising, since, as seen above, the pair method will disclose more information -- information which contributes to inference. In the situation when the true voting intentions are uniformly distributed with $N = 10$ Table \ref{tab: theoretical variance bias} implies that the variance for the list method is 4 times higher than for the pair method and we see that the ratio between the baseline variance w.r.t.\ $\tilde p_i$ and the pair method is 0.44 whereas the same ratio for the list method is only 0.11. The latter corresponds to that the pair method's variance is approximately 2 times higher than the baseline variance and the corresponding figure for the list method is 9 times higher.

When we consider the Swedish general election from 2014 in more detail we will focus on the Sweden Democrats (SD -- sensitive), the the Social Democrats (S -- largest) and those voting for ``other'' (O) less established parties (smallest). First, note that since the Swedish general election contains $N = 10$ different choices the list method will have constant variance --- which is not the case for the pair method and the benchmark binomial method. In Figures \ref{s_val_sd}, \ref{s_val_s} and \ref{s_val_ovr} we see that for the list method we need to have a sample size of approximately $n = 8~ 000$ in order to have a standard deviation of $1 \%$, i.e.\ a $95 \%$ confidence interval with width of $3.9 \%$, and need a sample size of approximately $n = 15~ 000$ in order to reduce the standard deviation to $0.75 \%$ ($95 \%$ confidence width of $2.9 \%$). Also note that the standard deviation for the list method is the same in all three figures. From the same figures we see that we need approximately $n = 5~ 000 - 12~ 000$ in order for the pair method to have a standard deviation of $0.5 \%$, noting that the standard deviation for the pair method (and the baseline binomial) differs between the figures.

Further, if we turn to calculating power, $\pi_\gamma(b_i)$, following \eqref{eq : power-bias}, we need to decide on the total survey size together with how many we should allocate to the standard direct poll and how many which should be allocated to answer according to the list or the pair method. In all analyses we use a total survey size of $n = 15~ 000$ chosen based on the above variance discussions and use the two alternatives for allocation: (i)  $n_{\mathit{List}} = n_{\mathit{Pair}} = 13~ 500$ and (ii) use allocations which are optimised w.r.t.\ $\pi_\gamma(b_i)$ under the null-hypothesis $b_i = 0$, that is we choose $n_\bullet$ according to
\[
n_\bullet = n \frac{\sqrt{\Var_{\hat p}[\hat p]}}{\sqrt{\Var_{\hat p}[\hat p]} + \sqrt{\Var_{\tilde p}[\tilde p]}},
\]
where $n_\bullet$ corresponds to either the list or the pair method. Moreover, in all calculations we have used the confidence level (type I error) $\gamma = 5 \%$.

By comparing Figure \ref{bias_equal_p_nbin1500} with Figure \ref{bias_equal_p_optim_n}, and Figure \ref{bias_sd_nbin1500} with Figure \ref{bias_sd_optim_n}, we see that there is a substantial gain in power when using the optimised allocation between the types of surveys. At $90 \%$ power the difference in detectable bias in the support for SD is approximately 1 percentage point for the pair method and 0.5 percentage points for the list method. For the optimised allocation (Figure \ref{bias_sd_optim_n}) it is possible to detect bias which is slightly less than 2 percentage points with the pair method and 3 percentage points with the list method, at $90 \%$ power. In the days before the Swedish general election in 2014 most large opinion polls underestimated the SD support by 2 to 3 percentage points, see e.g. \cite{svt2014alla}.

\section{Concluding remarks}

The present paper is concerned with methods on how to openly ask multiple choice questions where at least one alternative is seen as sensitive and where a response does not fully divulge the respondent's true position to the question asked. In order to do so we have presented two methods, the ``pair method’’ and the ``list method’’, which both rely on the idea that a response is defined in terms of a subset of all possible choices which still contains the respondent's true choice. This provides the respondent with plausible deniability if the sensitive answer(s) is contained in the subset of choices which defines its response. The degree of anonymity which is obtained in this way is possible to quantify in terms of both entropy related measures of privacy and in terms of ``jeopardy’’ measures. Using these measures it is possible to communicate the degree of anonymity which is retained when participating in this type of survey.

Further, for both methods we have derived unbiased estimators of the true underlying population proportions which belong to the different categories of choices, together with expressions for the estimators’ variances as well as central limit theorems. Moreover, these results also allows us to make power calculations w.r.t.\ detection of possible bias in responses when compared with standard direct survey methods.

As discussed in the introduction, our main motivation for this research topic is the recently observed problems in estimating voting intentions in general elections. In particular, the problem of correctly estimating the population proportion which will vote for populist parties/candidates. Due to this context, we believe that it is crucial that an anonymised survey method, apart from allowing for an anonymised responses, does not rely on complicated randomisation devices, is based on simple questions and allows for automated survey procedures. The methods introduced in the present paper allows for this. Moreover, by providing anonymity we believe that the introduced methods should be able to reduce response bias. Still, the reduction in the estimators’ precision, due to anonymisation, is possible to compensate for by increasing the sample size, which could be done inexpensively using automation. It is also worth to stress that both methods are based on random sampling — the use of non-randomised sampling techniques is one of the explanations for poor performance given in the post-election analyses given in \cite{sturgis2017assessment,sturgis2016report}.

\section*{Acknowledgement}

The first author is grateful for the support by AFA Insurance. Opinions expressed in this paper are not necessarily those of AFA Insurance.

\bibliographystyle{plain}
\bibliography{referenser.bib}

\appendix

\section{Proofs}\label{sec:proofs}

Let $\bm{1}$ be a column vector of ones whose length depends on the context, $J_{M, N}$ be an $M\times N$ matrix of ones, and let $M_n(a,b)$ be the $n\times n$ matrix with all diagonal elements equal to $a$ and all off-diagonal elements equal to $b$: $M_n(a,b) := (a-b)I + bJ_{n,n}$. We will sometimes write $A^+$ for $(A'A)^{-1}A$ since the latter is the Moore-Penrose inverse of $A$ when $A$ has full column rank.

\begin{proof}[Proof of Lemma \ref{lemmaUB}]
We first note that $A$ has dimensions $(\sum_i K_i) \times N$ so full (column) rank means having rank $N$ and in particular $(A'A)^{-1}$ is well-defined. Since  the $\bmX_i$'s have expected value $\Ex[\bmX_i] = n_i\bmu_i = n\alpha_i A_i\bmp$ and are independent with covariance matrices $nV(A_i\bmp) = n\alpha_iV(A_i\bmp)$ we have
\begin{align*}
\Ex[\bbmX] &= nA\bmp \\
\Var[\bbmX] &= n\begin{pmatrix} \alpha_1 V(A_1\bmp) && 0 \\ & \ddots & \\ 0 && \alpha_L V(A_L\bmp) \end{pmatrix}
\end{align*}
so that
\begin{align*}
\Ex[\hat{\bmp}] &= \frac{1}{n}(A'A)^{-1}A'\Ex[\bbmX] = (A'A)^{-1}A'A\bmp = \bmp, \\
\Var[\hat{\bmp}] &= \frac{1}{n^2}A^+\Var[\bbmX] A^{+\prime} \\
&= \frac{1}{n}(A'A)^{-1}A'\begin{pmatrix} \alpha_1 V(A_1\bmp) && 0 \\ & \ddots & \\ 0 && \alpha_L V(A_L\bmp) \end{pmatrix}A(A'A)^{-1} \\
&= \frac{1}{n}(A'A)^{-1}\bigg(\sum_{i=1}^L \alpha_i^3 A_i'V(A_i\bmp)A_i\bigg)(A'A)^{-1}.
\end{align*}
In the case $L=1$, the variance of the estimator is
\begin{align}
\Var[\hat\bmp] &= \frac1n (A'A)^{-1}A'V(A\bmp)A(A'A)^{-1} = \frac1n\big((A'A)^{-1}A'\diag(A\bmp)A(A'A)^{-1} - \bmp\bmp'\big). \label{L1_var}
\end{align}
\end{proof}

\begin{proof}[Proof of Lemma \ref{lemma: CGS}]
We want to prove a central limit theorem for
\[
\hat{\bm{p}} = \frac1n (A'A)^{-1}A'\bbmX,
\]
and will start by proving a central limit theorem for $\bbmX$. Recall that $\bmX_i \sim \Mult(n_i,\bmu_i)$ which gives us that
\[
\hat\bmu_i := \frac{1}{n_i}\bmX_i
\]
is an unbiased estimator of $\bmu_i = A_i\bmp$ with covariance $V(A_i\bmp)/n_i$. From e.g.\ \cite[Thm.\ 14.3-4]{bishop2008discrete} it then follows that
\[
\sqrt{n_i}(\frac{1}{n_i}\bmX_i - A_i\bmp) = \sqrt{n_i}(\hat\bmu_i - \bmu_i) \stackrel{D}{\to} \mathbf{U}_i,\quad \text{as } n_i\to\infty,
\]
where $\mathbf{U}_i$ is multivariate Gaussian with mean $\mathbf{0}$ and covariance $V(A_i\bmp)$. Thus, if
\[
\alpha_i = \frac{n_{i}}{n}\to \tilde\alpha_i\ge 0, \quad \text{as } n_i,n\to\infty,
\]
for all $i$ such that $\sum_i \tilde\alpha_i = 1$ it follows, due to Slutsky's theorem, that
\[
\sqrt{n}(\frac1n\bmX_i - \alpha_iA_i\bmp) = \sqrt{\alpha_i}\sqrt{n_i}(\hat\bmu_i - \bmu_i) \stackrel{D}{\to} \tilde{\mathbf{U}}_i, \quad \text{as } n_i,n\to\infty,
\]
where $\tilde{\mathbf{U}}_i$ is multivariate Gaussian with mean $\mathbf{0}$ and covariance $\tilde{\alpha}_iV(A_i\bmp)$. Hence, if we stack all $\bmX_i$ in $\bbmX$ as before, it follows, due to independence between the $\bmX_i$s, that
\[
\sqrt{n}(\frac1n\bbmX_i - A\bmp) \stackrel{D}{\to} \tilde{\mathbf{U}}, \quad \text{as } n_i,n\to\infty, \text{for all } i,
\]
where $\tilde{\mathbf{U}}_i$ is multivariate Gaussian with mean $\mathbf{0}$ and covariance
\[
\Var[\tilde{\bm{U}}] = \begin{pmatrix} \tilde{\alpha}_1 V(A_1\bmp) && 0 \\ & \ddots & \\ 0 && \tilde{\alpha}_L V(A_L\bmp), \end{pmatrix}
\]
by the same arguments as those used in the proof of Lemma \ref{lemmaUB}. Thus, we have obtained a central limit theorem for $\bbmX$.

Continuing, note that we may re-write
\[
\hat{\bm{p}} = \frac1n (A'A)^{-1}A'\bbmX = h_n(\frac1n\bbmX)
\]
where
\[
h_n(x) := (A'A)^{-1}A'x \in \mathbb{R}^{N\times 1}
\]
for $x \in \mathbb{R}^{(\sum_i K_i)\times 1}$. Moreover, given that
\[
\alpha_i = \frac{n_{i}}{n}\to \tilde\alpha_i\ge 0, \quad \text{as } n_i,n\to\infty,
\]
for all $i$ such that $\sum_i \tilde\alpha_i = 1$ it is clear that
\[
A \to \tilde{A}, \quad \text{as } n_i,n\to\infty, \quad \text{for all } i,
\]
and that
\[
(A'A)^{-1}A' \to (\tilde{A}'\tilde{A})^{-1}\tilde{A}'.
\]
Consequently, it follows that
\[
h_n(x) \to h(x), \quad \text{as } n_i,n\to\infty, \quad \text{for all } i,
\]
which together with \cite[Thm.\ 3.27]{kallenberg1997foundations} yields
\[
\sqrt{n}(\hat{\bm{p}}-\bm{p}) = h_n(\sqrt{n}(\frac1n\bbmX - A\bmp)) \stackrel{D}{\to} h(\tilde{U}) = (\tilde{A}'\tilde{A})^{-1}\tilde{A}'\tilde{U},
\]
where $\bm{G}:= (\tilde{A}'\tilde{A})^{-1}\tilde{A}'\tilde{U}$ is multivariate Gaussian with mean $\bm{0}$ and covariance
\[
\Var[\bm{G}] = (\tilde{A}'\tilde{A})^{-1}\bigg(\sum_{i=1}^L \tilde{\alpha}_i^3 A_i'V(A_i\bmp)A_i\bigg)(\tilde{A}'\tilde{A})^{-1},
\]
which concludes the proof.
\end{proof}

\begin{proof}[Proof of Theorem \ref{thm_pair}]
As noted above, we can write $\bmu = A\bmp$, where $A=\frac{1}{N-1}B'$ and $B$ is the $N\times M$ matrix with elements $b_{ik}$ as defined above:
$$
A \coloneqq \frac{1}{N-1}\begin{pmatrix} 1 & 1 & 0 & \dots & 0 & 0 \\
                                         1 & 0 & 1 & \dots & 0 & 0\\
                                         \vdots & \vdots & \vdots & & \vdots & \vdots \\
                                         0 & 0 & 0 & \dots & 1 & 1 \end{pmatrix}
$$

Since $A$ has rank $N$, we can use Lemma \ref{lemmaUB} with $L=1$ to derive an unbiased estimator of $\bmp$, and we will in fact show that the $\hat\bmp$ of the Theorem equals that of the Lemma, i.e.\ that $\hat\bmp$ of the Theorem equals $A^+\hat\bmu$. We claim
\begin{equation}\label{aplus}
A^+ = \frac{N-1}{N-2}B - \frac{1}{N-2}J_{N,M}.
\end{equation}
To prove \eqref{aplus}, it is sufficient to show that $A^+A = I$. First note that $J_{K,M}A = J_{K,N}$ since $A\in\mathcal{A}(M,N)$ and thus $\bm{1}'A = \bm{1}'$. Also note that $BB'$ is a matrix with all diagonal elements equal to $N-1$ and all off-diagonal elements equal to 1, i.e.\ 
\begin{equation}\label{BB}
BB'=(N-2)I+J_{N,N}.
\end{equation}
Therefore,
\begin{align*}
A^+A &= \left(\frac{N-1}{N-2}B - \frac{1}{N-2}J_{N,M}\right)A \\
&= \frac{N-1}{N-2}BA - \frac{1}{N-2}J_{N,N} \\
&= \frac{1}{N-2}BB' - \frac{1}{N-2}J_{N,N} \\
&= I,
\end{align*}
and we have thus shown that \eqref{aplus} holds. Since $\hat\bmu$ is a vector of probabilities, $J_{N,M}\hat\bmu = \bm{1}$, and
$$
A^+\hat\bmu = \left(\frac{N-1}{N-2}B - \frac{1}{N-2}J_{N,M}\right)\hat\bmu = \frac{N-1}{N-2}B\hat\bmu - \frac{1}{N-2}\bm{1} = \hat\bmp,
$$
where the right hand side is defined in \eqref{phat_par}. By Lemma \ref{lemmaUB}, $\hat\bmp$ is an unbiased estimator of $\bmp$.

It remains to show that \eqref{L1_var} specializes to \eqref{phat_par_var} and \eqref{phat_par_cov} in our case. We see that it suffices to show that $A^+\diag(A\bmp)A^{+\prime}$ has diagonal elements $(1+(N-3)p_i)/(N-2)$ for $i=1,\dots,N$ and off-diagonal elements $-(1-p_i-p_j)/(N-2)^2$.

In general, when a matrix $E$ with elements $e_{ij}$, $i,j=1,\dots,N$, is defined as the product $E \coloneqq C\diag(\bm{d})C'$ where $C$ has elements $c_{ij}$ for $i=1,\dots,N$, $j=1,\dots,M$ and $\bm{d} = (d_1,\dots,d_M)'$, then
\begin{equation}\label{c_diag_c}
e_{ij} = \sum_k d_k c_{ik}c_{jk}. 
\end{equation}
We use equation \eqref{c_diag_c} with $\bm{d} = A\bmp = \frac{1}{N-1}(p_1+p_2,\dots,p_{N-1}+p_N)' = \frac{1}{N-1}B'\bmp$ and
$$
C = A^+ = \frac{1}{N-2}\big((N-1)B-J_{N,M}\big).
$$
Equation \eqref{aplus} means that $A^+$ has elements $(A^+)_{ik} = \frac{1}{N-2}\big((N-1)b_{ik}-1\big) = 1$. Element $k$ of $A\bmp = \frac{1}{N-1}B'\bmp$ is $\frac{1}{N-1}\sum_l p_lb_{lk}$.

For the diagonal elements we get
\begin{align*}
\big(A^+\diag(A\bmp)A^{+\prime}\big)_{ii} &= \sum_k d_k c_{ik}^2 \\
&= \frac{1}{(N-1)(N-2)^2}\sum_{k,l}p_l b_{lk}\big((N-1)b_{ik}-1\big)^2 \\
\{b_{ik}^2 = b_{ik}\}\quad &= \frac{1}{(N-1)(N-2)^2}\sum_{k,l}p_lb_{lk}\big((N-1)(N-3)b_{ik} + 1\big) \\
&= \frac{N-3}{(N-2)^2}\sum_{k,l}p_lb_{lk}b_{ik} + \frac{1}{(N-1)(N-2)^2}\sum_{k,l}p_lb_{lk} \\
\{\text{by \eqref{BB} and def.\ of $B$}\}\quad &= \frac{N-3}{(N-2)^2}\Big((N-2)p_i+\sum_l p_l\Big) + \frac{1}{(N-2)^2}\sum_{l}p_l \\
&= \frac{1+(N-3)p_i}{N-2},
\end{align*}
which is what we want. We proceed with the off-diagonal elements with $i\neq j$.
\begin{align*}
\big(A^+\diag(A\bmp)A^{+\prime}\big)_{ij} &= \sum_k d_k c_{ik}c_{jk} \\
&= \frac{1}{(N-1)(N-2)^2}\sum_{k,l}p_l b_{lk}\big((N-1)b_{ik}-1\big)\big((N-1)b_{jk}-1\big) \\
&= \frac{N-1}{(N-2)^2}\sum_{k,l}p_l b_{lk}b_{ik}b_{jk} -  \frac{1}{(N-2)^2}\sum_{k,l}p_l b_{lk}b_{ik} \\
&\qquad -\frac{1}{(N-2)^2}\sum_{k,l}p_l b_{lk}b_{jk} + \frac{1}{(N-1)(N-2)^2}\sum_{k,l}p_l b_{lk} \\
&= \frac{N-1}{(N-2)^2}(p_i+p_j) -  \frac{1}{(N-2)^2}\big((N-2)p_i+1\big) \\
&\qquad -\frac{1}{(N-2)^2}\big((N-2)p_j+1\big) + \frac{1}{(N-2)^2} \\
&= -\frac{1-p_i-p_j}{(N-2)^2},
\end{align*}
and the proof is done.
\end{proof}

\begin{proof}[Proof of Proposition \ref{prop:lika_varians}]
We will apply Lemma \ref{lemmaUB} and make the expression in \eqref{phat_var2} more explicit, at least with regards to the diagonal of the covariance matrix. Note that 
\begin{equation*}
M_n(a,b)M_n(c,d) = M_n\big(ac + bd(n-1), ad + bc + bd(n-2)\big)
\end{equation*}
and thus
\begin{equation}\label{Minv}
M_n(a,b)^{-1} = M_n\left(\frac{a + b(n-2)}{(a-b)(a + b(n-1))}, \frac{-b}{(a-b)(a + b(n-1))}\right)
\end{equation}
provided $a\neq b$ and $a \neq -b(n-1)$. For the special case $J_{n,n} = M_n(1,1)$ we have
\begin{equation}\label{Jmult}
M_n(a,b)J_{n,n} = M_n(a+b(n-1),a+b(n-1)) = (a+b(n-1))J_{n,n}.
\end{equation}

With $A$ the matrix defined in Lemma \ref{lemmaUB}:
$$
A \coloneqq \alpha\begin{pmatrix} A_1 \\ \vdots \\ A_L \end{pmatrix} \eqqcolon \alpha \bar A.
$$
The $2L \times N$ matrix $\bar A$ codes with ones for membership of a party in each of the $2L$ lists and complementary lists, and has zeros for non-membership. 

The matrix $A'A = \alpha^2 \bar{A}'\bar{A} \eqqcolon \alpha^2 B$ where the elements of $B=(b_{ij})$ count the number of lists and complementary lists that include both party $i$ and $j$. Since the lists and complementary lists include all combinations of $N/2$ out of $N$ parties and each party is included in half of the combinations, $b_{ii} = \binom{N}{N/2}/2 = L$ and $b_{ij} = \binom{N-2}{N/2-2} = L\frac{N/2-1}{N-1}$ for $i\neq j$. In other words, we have
$$
A'A = \alpha^2 M_N\bigg(L,L\frac{N/2-1}{N-1}\bigg) = \alpha M_N\bigg(1,\frac{N/2-1}{N-1}\bigg)
$$
with inverse, by \eqref{Minv},
$$
(A'A)^{-1} = 2L M_N\bigg(1-\frac{2(N-1)}{N^2}, -\frac{N-2}{N^2}\bigg).
$$
Observe that, by equation \eqref{Jmult} and the matrices being symmetric,
\begin{equation}
(A'A)^{-1}J_{N,N} = J_{N,N}(A'A)^{-1} = \frac{1}{N}J_{N,N}.
\end{equation}
Since $A_i\bmp = (p^+_i,p^-_i)'$, where $ p_i^\circ \coloneqq \sum_{k\in\mathfrak{L}_i^\circ}p_k$ for $\circ = +,-$, we have
\begin{align*}
V(A_i\bmp) &= p^+_ip^-_i\begin{pmatrix} 1 & -1 \\ -1 & 1\end{pmatrix} \eqqcolon q_i \begin{pmatrix} 1 & -1 \\ -1 & 1\end{pmatrix} = q_i (2I - J_{2,2}),
\end{align*}
where we defined $q_i \coloneqq p_i^+p_i^-$. Now
$$
A_i'V(A_i\bmp)A_i=q_i(2A_i'IA_i - A_i'J_{2,2}A_i) = q_i(2A_i'A_i - J_{N,N}) = q_i C_i,
$$
where the matrix $C_i$ has equal number of 1's and $-1$'s on each row and column. For such a matrix, $C_iJ_{N,N} = J_{N,N}C_i = \mathbf{0}$, where $\mathbf{0} \coloneqq 0J_{N,N}$ is the zero matrix, and thus
\begin{align*}
M_N(a,b)C_iM_N(a,b) &= M_n(a,b)C_i\big((a-b)I + bJ_{N,N}\big) = (a-b)M_N(a,b)C \\
&= (a-b)\big((a-b)I+J_{N,N}\big)C_i = (a-b)^2C_i.
\end{align*}
We can write equation \eqref{phat_var2}
\begin{align*}
\Var[\hat{\bmp}] &= \frac{1}{n}(A'A)^{-1}\bigg(\sum_{i=1}^L \alpha_i^3 A_i'V(A_i\bmp)A_i\bigg)(A'A)^{-1} \\
&= \frac{4}{nL}\sum_{i=1}^Lq_iM_N\bigg(1-\frac{2(N-1)}{N^2}, -\frac{N-2}{N^2}\bigg)C_iM_N\bigg(1-\frac{2(N-1)}{N^2}, -\frac{N-2}{N^2}\bigg) \\
&= \frac{4}{nL}\bigg(1-\frac{1}{N}\bigg)^2\sum_{i=1}^Lq_iC_i
\end{align*}
Since the diagonal of each $C_i$ equals $\bm{1}$ all diagonal elements of the covariance matrix $\Var[\hat{\bmp}]$ equal $\frac{4}{nL}(1-\frac{1}{N})^2\sum_{i=1}^Lq_i$.
\end{proof}

\begin{proof}[Derivation of the variance and covariance entries in table \ref{tab: theoretical variance bias}]

Since $A_i\bmp = (\frac12,\frac12)'$ for all $i$, we have
$$
V(A_i\bmp)=\frac14 \begin{pmatrix*}[r] 1 & -1 \\ -1 & 1 \end{pmatrix*}=\frac14(2I-J_{2,2})
$$
and
$$
A_i'V(A_i\bmp)A_i=\frac14(2A_i'IA_i - A_i'J_{2,2}A_i) = \frac14(2A_i'A_i - J_{N,N}),
$$
and the sum that appears in equation \eqref{phat_var2} is
\begin{align*}
\sum_{i=1}^L\alpha_i^3A_i'V(A_i\bmp)A_i &= \sum_{i=1}^L\frac{\alpha^3}{4}(2A_i'A_i - J_{N,N}) \\
&= \frac{\alpha^3}{2}\sum_{i=1}^LA_i'A_i - \frac{\alpha^3 L}{4}J_{N,N}   \\
&= \frac{1}{2L}A'A - \frac{1}{4L^2}J_{N,N}.
\end{align*}
Putting it all together we get
\begin{align*}
\frac{1}{n}(A'A)^{-1}&\bigg(\sum_{i=1}^L \alpha_i^3 A_i'V(A_i\bmp)A_i\bigg)(A'A)^{-1} = \frac{1}{n}(A'A)^{-1}\bigg(\frac{1}{2L}A'A - \frac{1}{4L^2}J_{N,N}\bigg)(A'A)^{-1} \\
&= \frac{1}{n}\bigg(\frac{1}{2L}(A'A)^{-1} - \frac{1}{4L^2}(A'A)^{-1}J_{N,N}(A'A)^{-1}\bigg) \\
&= \frac{1}{n}\bigg(M_N\bigg(1-\frac{2(N-1)}{N^2}, -\frac{N-2}{N^2}\bigg) - \frac{1}{N^2}J_{N,N}\bigg) \\
&= \frac{1}{n}M_N\bigg(1-\frac{2N-1}{N^2}, -\frac{N-1}{N^2}\bigg) \\
&= \frac{1}{n}M_N\bigg(\Big(1-\frac1N\Big)^2, -\frac1N\Big(1-\frac1N\Big)\bigg).
\end{align*}
\end{proof}

\section{Entropy calculations: Pair and list method}

As stated in Section \ref{sec:entropy} the overall entropy of the distribution of voting intentions is
$$
H[T] = -\Ex_T[\log_2 p_T(T)] = -\sum_{i=1}^N p_i \log_2 p_i.
$$
For the pair method it holds that
$$
p_{T,R}(i, \{i,j\}) = \Prob(T = i, R=\{i,j\})=\frac{p_i}{N-1}
$$
and
$$
p_{T\mid R}(i\mid \{i,j\}) = \Prob(T=i\mid R=\{i,j\}) = \frac{p_i}{p_i+p_j}
$$
for all $i\neq j$. Further, by using the above it follows that
\begin{equation}\label{bet_ent_par}
H[T\mid R] = -\Ex_{T,R}[\log_2 p_{T|R}(T\mid R)] = -\sum_{i\neq j}\frac{p_i}{N-1}\log_2 \frac{p_i}{p_i+p_j}
\end{equation}
and
\begin{equation}\label{info_par}
I[T;R] = H[T] - H[T\mid R] = -\sum_{i\neq j} \frac{p_i}{N-1}\log_2 (p_i + p_j).
\end{equation}

For the list method, let $\mathfrak{L}_l^\circ$ denote either $\mathfrak{L}_l^+$ or $\mathfrak{L}_l^-$, and let $ p_l^\circ \coloneqq \sum_{k\in\mathfrak{L}_l^\circ}p_k$.
\begin{equation*}
p_{T,R}(i, \mathfrak{L}_l^\circ) = p_i\alpha_l
\end{equation*}
if $i\in\mathfrak{L}_l^\circ$, and
$$
p_{T|R}(i\mid \mathfrak{L}_l^\circ) = \frac{p_i}{ p_l^\circ}.
$$
Thus
\begin{align}
H[T\mid R] &= -\Ex_{T,R}[\log_2 p_{T\mid R}(T\mid R)] = -\sum_{l=1}^L\sum_{\substack{ i\in\mathfrak{L}_l^\circ \\ \circ = +,-}}
p_i\alpha_l\log_2 \frac{p_i}{ p_l^\circ} \notag\\
&= H[T] +  \sum_{l=1}^L\alpha_l( p_l^+\log_2 p_l^+ +  p_l^- \log_2 p_l^-)   \label{bet_ent_list}
\end{align}
and
\begin{equation}\label{info_list}
I[T;R] = H[T] - H[T\mid R] = -\sum_{l=1}^L\alpha_l( p_l^+\log_2 p_l^+ +  p_l^- \log_2 p_l^-).
\end{equation}
Note that for both the pair method and the list method
$$
p_{T|R}(i|r) = \frac{p_i}{p_i + \sum_{j\in r\setminus\{i\}}p_j}
$$
For a given $i$ this conditional probability is maximized when the other parties indicated in the response $r$ have the least possible support, e.g.\ when $i$ by chance is paired with the smallest other party in the pair method.

\section{Jeopardy calculations: Pair and list method}

For the pair method $\Prob(R=\{1,j\}\mid T=1) = \frac{1}{N-1}$ and
$$
\Prob(R=\{1,j\}\mid T\neq 1) = \frac{\Prob(R=\{1,j\},T\neq 1)}{\Prob(T\neq 1)} = \frac{\frac{1}{N-1}p_j}{1-p_1},
$$
so
$$
J(\{1,j\}) = \frac{1-p_1}{p_j},
$$
and $J(r) = 0$ when $1\notin r$. Averaging over the $\frac{N(N-1)}{2}$ possible responses yields
$$
\bar J = \frac{2(1-p_1)}{N(N-1)}\sum_{j\neq 1}\frac{1}{p_j}.
$$

The list method is similar with $\Prob(R=\mathfrak{L}^\circ_l\mid T = 1) = \alpha_l$ when $1\in\mathfrak{L}^\circ_l$ and
\begin{align*}
\Prob(R=\mathfrak{L}^\circ_l\mid T \neq 1) &= \frac{\Prob(R=\mathfrak{L}^\circ_l, T \neq 1)}{\Prob(T\neq 1)} = \frac{\alpha_l\Prob(T\in\mathfrak{L}^\circ_l\setminus\{1\})}{1-p_1} \\
&= \frac{\alpha_l}{1-p_1}\sum_{j\in\mathfrak{L}^\circ_l\setminus\{1\}}p_j = \alpha_l\frac{p^\circ_l-p_1}{1-p_1}.
\end{align*}
The jeopardy is thus
$$
J(\mathfrak{L}^\circ_l) = \frac{1-p_1}{p^\circ_l-p_1}\ett\{1\in\mathfrak{L}^\circ_l\}
$$
and averaging over all $2L$ lists and complementary lists we get
$$
\bar J = \frac{1-p_1}{2L}\sum_{\mathfrak{L}^\circ_l:1\in\mathfrak{L}^\circ_l}\frac{1}{p^\circ_l-p_1}.
$$

\newpage

\section{Tables}

In all tables ($\ast$) indicates that the the list method uses all lists with exactly half of the number of parties each, and each list has equal weight $\alpha_l$. To prevent double counting, party 1 is always on a list $\mathfrak{L}^+$ and never on a complementary list $\mathfrak{L}^-$.

\begin{table}[h]
\begin{tabular}{c|r|c|c|c|c|c|c|c|c|c}
\rule[-2mm]{0mm}{6mm} & \multicolumn{10}{|c}{Swedish election 2014}\\
\hline
\hline
\rule[-2mm]{0mm}{6mm}  & \textit{SD} & \textit{S} & \textit{M} & \textit{MP} & \textit{C} & \textit{V} & \textit{FP} & \textit{KD} & \textit{FI} & \textit{O}\\
\hline
\rule[-2mm]{0mm}{6mm} $p_i$ & 0.129 &  0.310 & 0.233 & 0.061 & 0.069 & 0.057 & 0.054 & 0.046 & 0.031 & 0.010 \\
\hline
\end{tabular}
\vspace{2mm}

\caption{Outcome of the Swedish general election from 2014.}\label{tab: SWE election}
\end{table}

\begin{table}[h]
\begin{tabular}{cr|c|c}
& \rule[-2mm]{0mm}{6mm} & General $\bmp$ & $p_i = \frac1N$ ($\ast$)\\
\hline
\hline
\rule[-2mm]{0mm}{6mm}
& $H[T]$ & $-\sum_ip_i\log_2p_i$ & $\log_2 N$\\
\hline
\parbox[t]{2mm}{\multirow{3}{*}{\rotatebox[origin=c]{90}{Pair}}} \rule[-2mm]{0mm}{6mm} & $I[T;R]$ & $-\sum_{i\neq j}\frac{p_i}{N-1}\log_2(p_i+p_j)$ & $\log_2 N -1$ \\ 
\rule[-2mm]{0mm}{6mm}  & $H[T|R]$ & $-\sum_{i\neq j}\frac{p_i}{N-1}\log_2\frac{p_i}{p_i+p_j}$ & 1\\
\rule[-2mm]{0mm}{6mm} & $-\max_r\log_2p_{T|R}(1|r)$ & $-\log_2p_1 + \log_2(p_1+\min_{j\neq 1}p_j)$ & 1\\
\hline
\parbox[t]{2mm}{\multirow{3}{*}{\rotatebox[origin=c]{90}{List}}} \rule[-2mm]{0mm}{6mm} & $I[T;R]$ & $-\sum_{l=1}^L\alpha_l( p_l^+\log_2 p_l^+ +  p_l^- \log_2 p_l^-)$ & 1 \\ 
\rule[-2mm]{0mm}{6mm}  & $H[T|R]$  &  $-\sum_ip_i\log_2p_i+  \sum_{l=1}^L\alpha_l( p_l^+\log_2 p_l^+ +  p_l^- \log_2 p_l^-)$ & $\log_2 N -1$ \\
\rule[-2mm]{0mm}{6mm} & $-\max_r\log_2p_{T|R}(1|r)$ & $-\log_2p_1 + \log_2\big(p_1+\min_{l,\circ;1\in\mathfrak{L}^\circ_l}\sum_{j\in\mathfrak{L}^\circ_l\setminus \{1\}}p_j\big)$ &  $\log_2 N -1$\\
\hline
\end{tabular}
\vspace{2mm}

\caption{Entropy related measures. $I[T; R]$ is the expected amount of divulged information, $H[T|R]$ is the expected amount of retained privacy, and $-\max_r\log_2 p_{T|R}(1|r)$ is the least possible amount of retained privacy for a respondent with embarrassing opinion (1).}\label{tab: theoretical entropy}
\end{table}

\begin{table}[h]
\begin{tabular}{cr|c|c}
\rule[-2mm]{0mm}{6mm} & & Swedish election 2014 ($\ast$)& $p_i = \frac{1}{10}$ ($\ast$)\\
\hline
\hline
\rule[-2mm]{0mm}{6mm} & $H[T]$ & 2.80 & 3.32\\
\hline
\parbox[t]{2mm}{\multirow{3}{*}{\rotatebox[origin=c]{90}{Pair}}} \rule[-2mm]{0mm}{6mm} & $I[T;R]$ & 2.06 & 2.32\\
 \rule[-2mm]{0mm}{6mm} & $H[T|R]$ & 0.74 & 1.00\\
 \rule[-2mm]{0mm}{6mm} & $-\max_r\log_2p_{T|R}(1|r)$ & 0.11 & 1.00\\ 
\hline
\parbox[t]{2mm}{\multirow{3}{*}{\rotatebox[origin=c]{90}{List}}}  \rule[-2mm]{0mm}{6mm} & $I[T;R]$ & 0.93 & 1.00 \\ 
 \rule[-2mm]{0mm}{6mm} & $H[T|R]$  &  1.87 & 2.32 \\
\rule[-2mm]{0mm}{6mm} & $-\max_r\log_2p_{T|R}(1|r)$ & 1.07 & 2.32\\ 
\hline
\end{tabular}
\vspace{2mm}

\caption{Entropy related measures --- the general Swedish election 2014.}\label{tab: SWE entropy}
\end{table}

\begin{table}[h]
\begin{tabular}{cr|c|c}
\rule[-2mm]{0mm}{6mm} & & General $\bmp$ & $p_i = \frac1N$ ($\ast$)\\
\hline
\hline
\parbox[t]{2mm}{\multirow{2}{*}{\rotatebox[origin=c]{90}{Pair}}} \rule[-2mm]{0mm}{6mm} & $J(r)$ & $\frac{1-p_1}{p_j}\ett\{1\in r\}$ & $(N-1)\ett\{1\in r\}$ \\ 
 \rule[-2mm]{0mm}{6mm} & $\bar J$ & $\frac{2(1-p_1)}{N(N-1)}\sum_{j\neq 1}\frac{1}{p_j}$ & $2\big(1-\frac1N\big)$\\
\hline
\parbox[t]{2mm}{\multirow{2}{*}{\rotatebox[origin=c]{90}{List}}} \rule[-2mm]{0mm}{6mm} & $J(\mathfrak{L}^\circ_l)$ & $\frac{1-p_1}{p^\circ_l-p_1}\ett\{1\in \mathfrak{L}^\circ_l\}$ & $\frac{2(N-1)}{N-2}\ett\{1\in \mathfrak{L}^\circ_l\}$ \\ 
 \rule[-2mm]{0mm}{6mm} & $\bar J$  &  $\frac{1-p_1}{2L}\sum_{\mathfrak{L}^\circ_l:1\in\mathfrak{L}^\circ_l}\frac{1}{p^\circ_l-p_1}$ & $\frac{N-1}{N-2}$ \\
\hline
\end{tabular}
\vspace{2mm}

\caption{Jeopardy measures.}\label{tab: theoretical jeopardy}
\end{table}

\begin{table}[h]
\begin{tabular}{cr|c|c}
\rule[-2mm]{0mm}{6mm} & & Swedish election 2014 ($\ast$)& $p_i = \frac{1}{10}$ ($\ast$)\\
\hline
\hline
\parbox[t]{2mm}{\multirow{2}{*}{\rotatebox[origin=c]{90}{Pair}}}     & \rule[-2mm]{0mm}{6mm}  $\max J(r)$ & 87.1 & 9.00\\
 & \rule[-2mm]{0mm}{6mm}  $\bar J$ & 4.42 & 1.80\\
\hline
\parbox[t]{2mm}{\multirow{2}{*}{\rotatebox[origin=c]{90}{List}}}      & \rule[-2mm]{0mm}{6mm}  $\max J(\mathfrak{L}^\circ_l)$  & 6.18  & 2.25 \\
& \rule[-2mm]{0mm}{6mm}  $\bar J$  &  1.37 & 1.13 \\
\hline
\end{tabular}
\vspace{2mm}

\caption{Jeopardy measures --- the general Swedish election 2014.}\label{tab: SWE jeopardy}
\end{table}

\begin{table}[h]
\begin{tabular}{l|c|c}
\rule[-2mm]{0mm}{6mm} $p_i = \frac{1}{N}$ ($\ast$) & $\Var(\hat p_i)$ & $\Cov(\hat p_i, \hat p_j)$ \\
\hline
\hline 
Pair & \rule[-2mm]{0mm}{6mm} $\frac{2}{n(N-2)}(1-\frac1N)^2$ & $-\frac{2}{nN(N-2)}(1-\frac1N)$\\
\hline
List &  \rule[-2mm]{0mm}{6mm} $\frac{1}{n}(1-\frac1N)^2$ & $-\frac{1}{nN}(1-\frac1N)$\\
\hline
Baseline & \rule[-2mm]{0mm}{6mm} $\frac{1}{nN}(1-\frac1N)$ & $-\frac{1}{nN^2}$\\
\hline
\end{tabular}
\vspace{2mm}

\caption{Variance and bias relations. The values for the pair method are from Theorem \ref{thm_pair} and the list expressions are deduced in Appendix \ref{sec:proofs}.}\label{tab: theoretical variance bias}
\end{table}

\clearpage

\section{Figures}

\begin{figure}[h]
\center
\includegraphics[width=11cm]{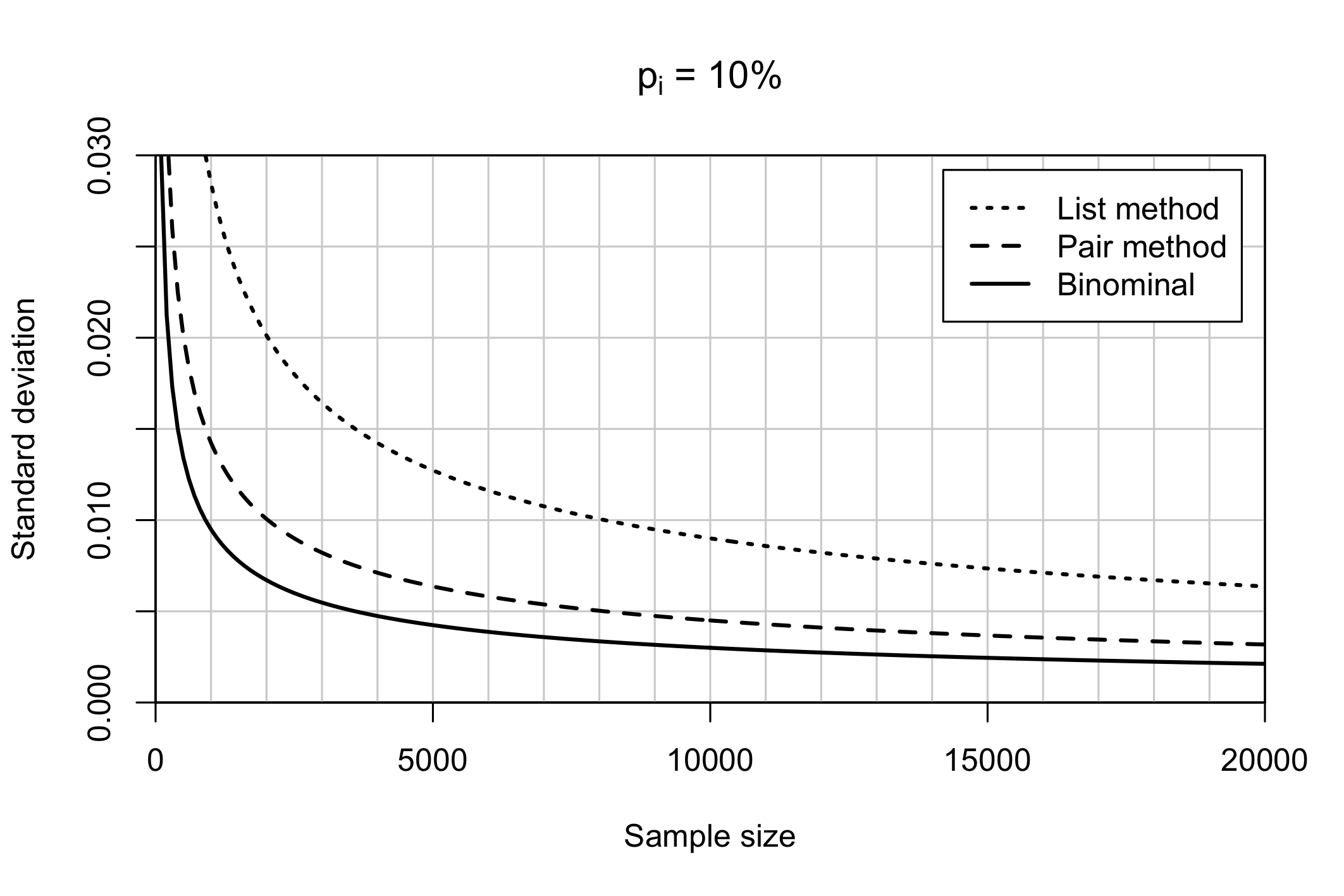}
\caption{Standard deviation for the estimate of a party's support as a function of sample size in the case of equal true supports $p_i = 10\%$ for $i=1,\dots,10$.}
\label{s_lika}
\end{figure}

\begin{figure}[h]
\center
\includegraphics[width=11cm]{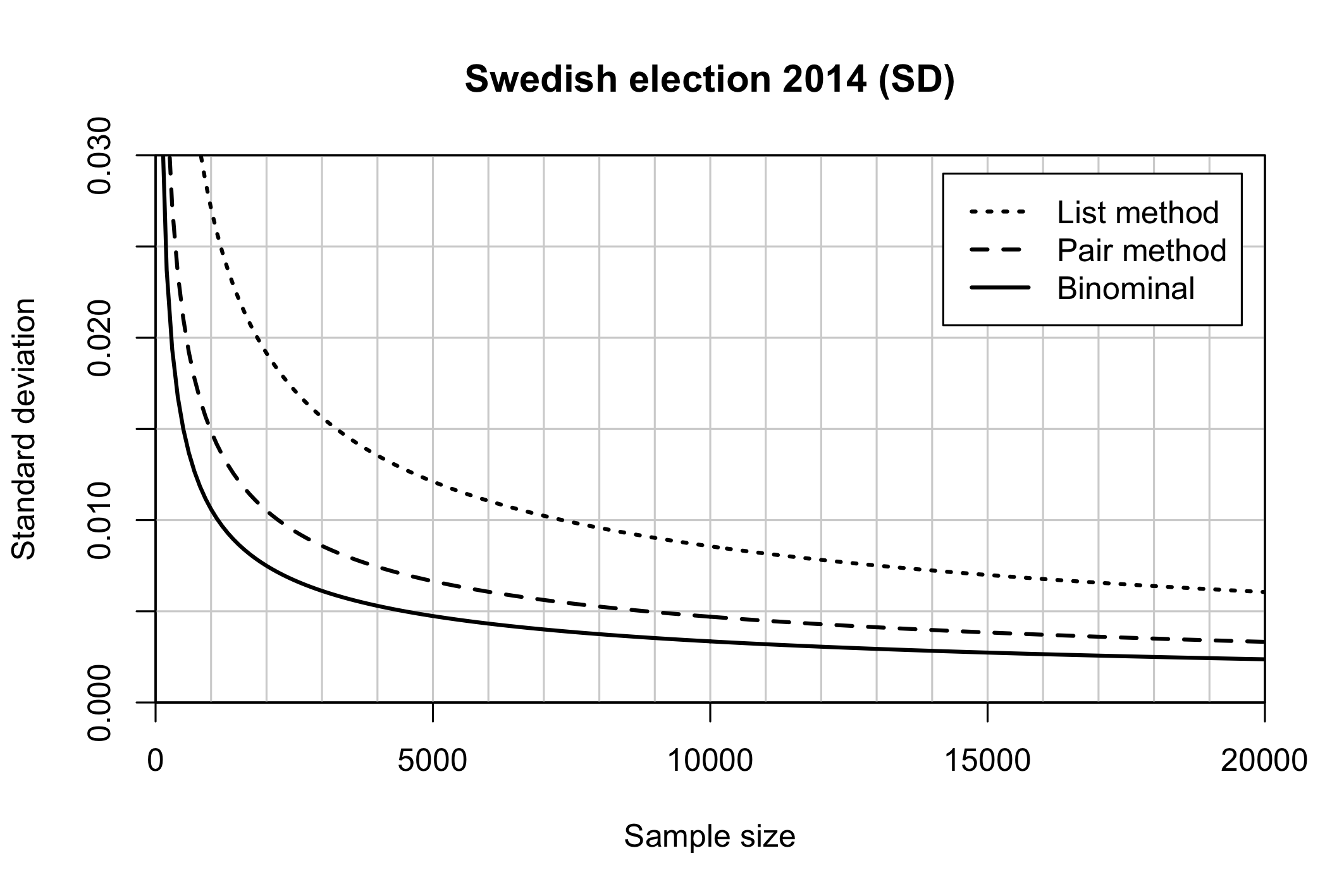}
\caption{Standard deviation for the estimate of Sweden Democrats' support as a function of sample size in the case of true support equal to the 2014 election (12.9\%).}
\label{s_val_sd}
\end{figure}

\begin{figure}[h]
\center
\includegraphics[width=11cm]{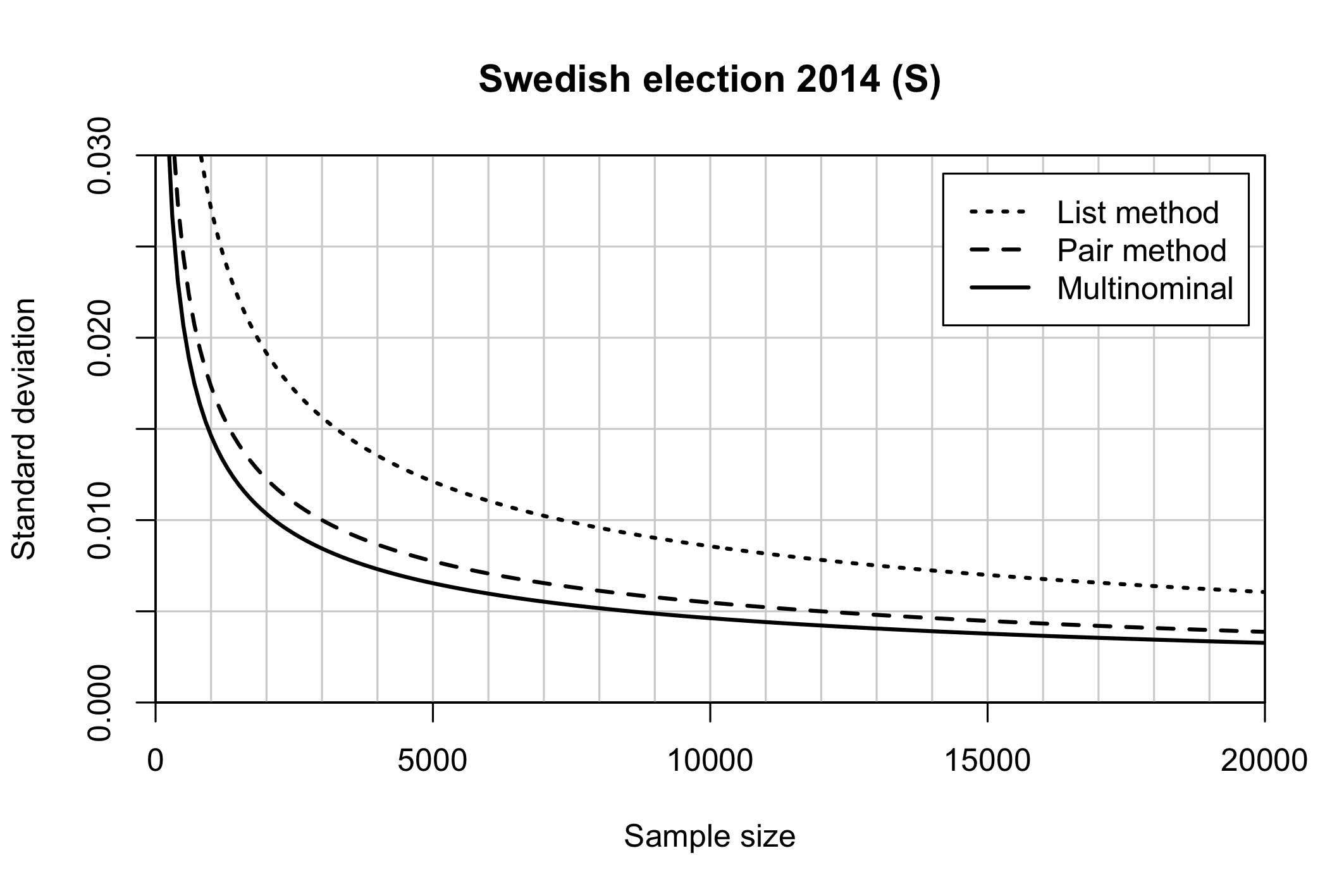}
\caption{Standard deviation for the estimate of Social Democrats' support as a function of sample size in the case of true support equal to the 2014 election (31.0\%).}
\label{s_val_s}
\end{figure}

\begin{figure}[h]
\center
\includegraphics[width=11cm]{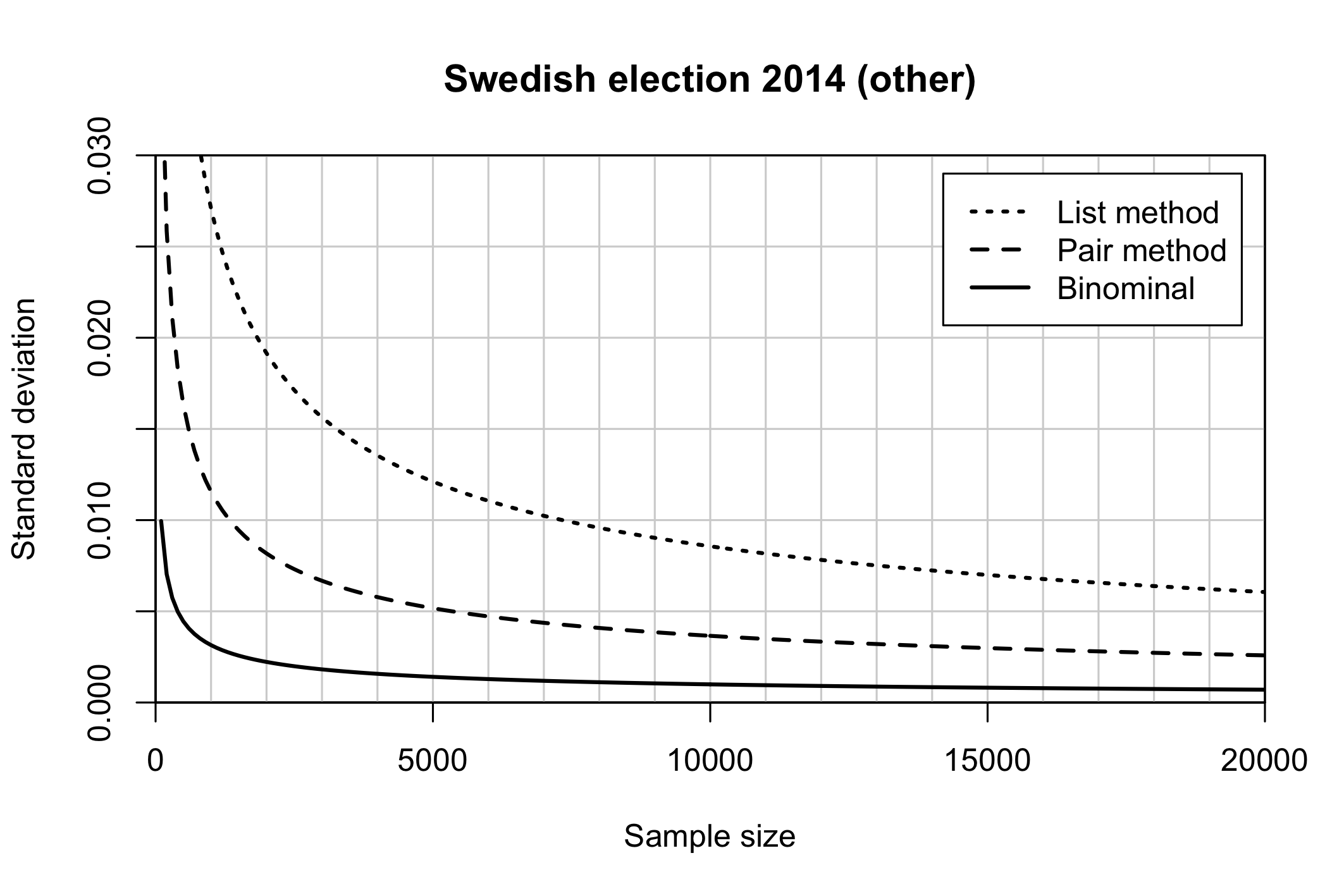}
\caption{Standard deviation for the estimate of minor parties' support as a function of sample size in the case of true support equal to the 2014 election (1\%).}
\label{s_val_ovr}
\end{figure}

\begin{figure}[h]
\center
\includegraphics[width=11cm]{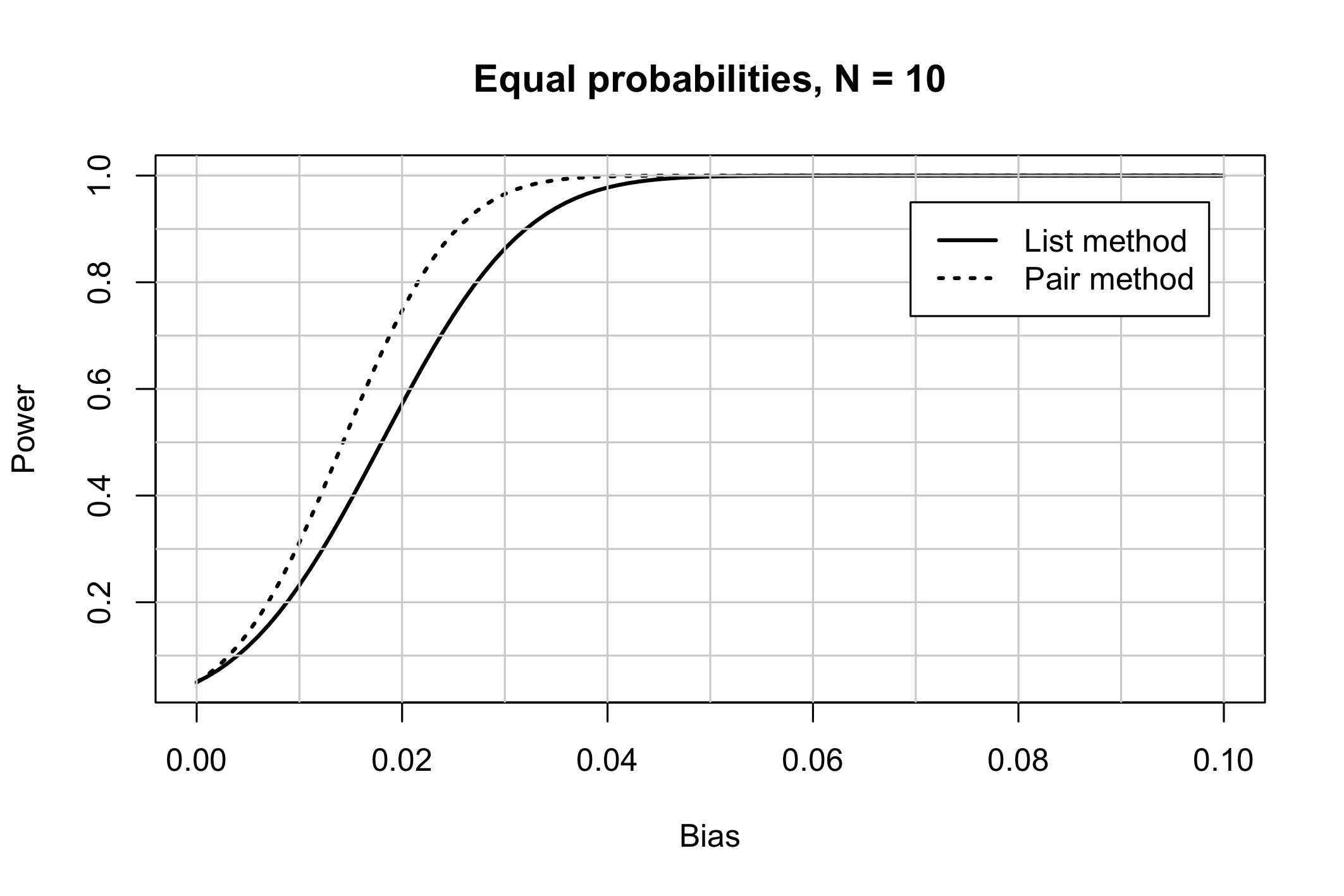}
\caption{Power calculation for bias-detection according to \eqref{eq : power-bias} where the total size of both surveys is $n = 15~ 000$ with $n_{\mathit{List}} = n_{\mathit{Pair}} = 13~ 500$, $p = 10 \%$, $N = 10$. The confidence level (type I error) used is $\gamma = 5 \%$.}
\label{bias_equal_p_nbin1500}
\end{figure}

\begin{figure}[h]
\center
\includegraphics[width=11cm]{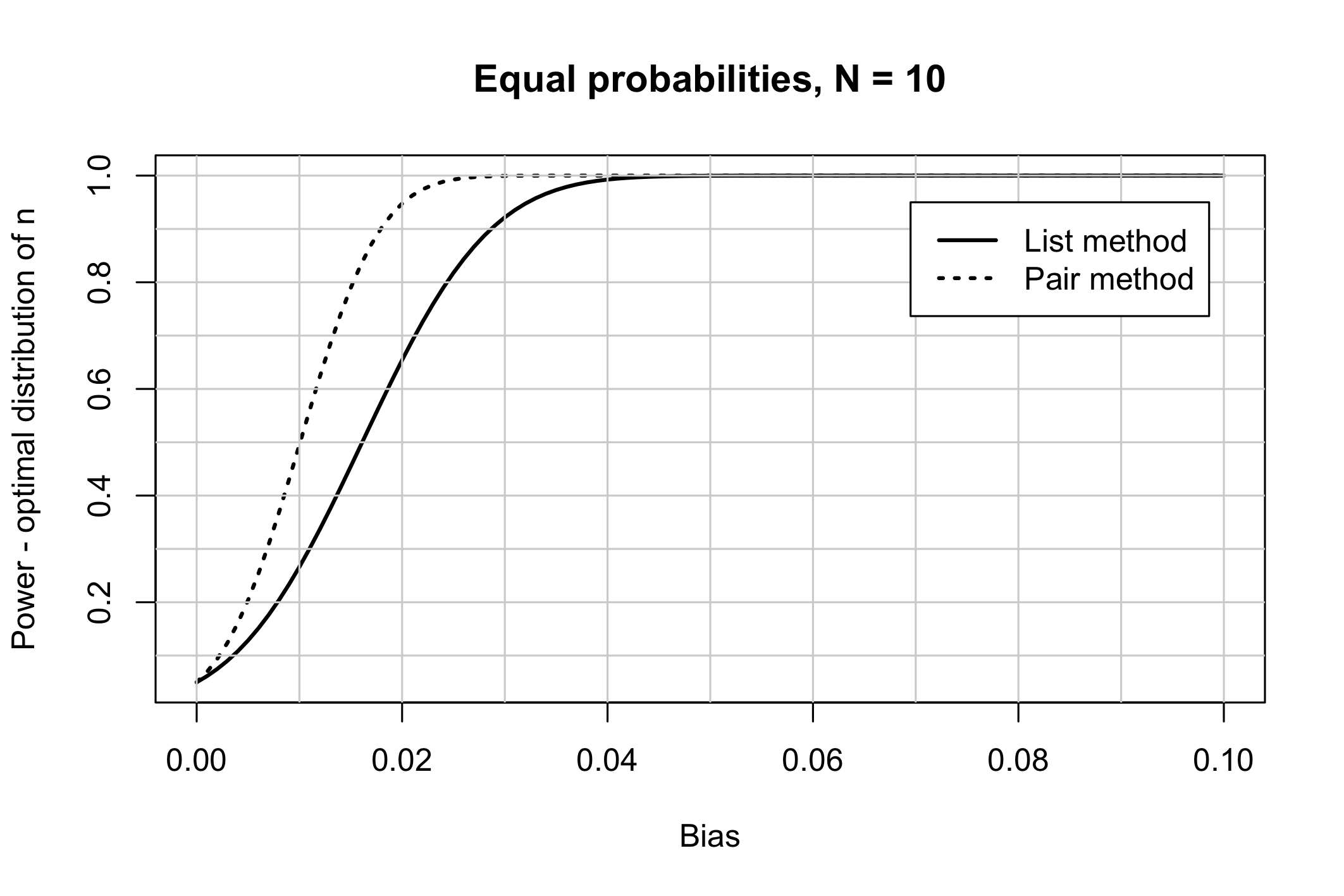}
\caption{Power calculation for bias-detection according to \eqref{eq : power-bias} where the total size of both surveys is $n = 15~ 000$, with optimised allocation between the standard binomial survey and the list and the pair method according to $n_{\mathit{List}} = 11~ 250$ and $n_{\mathit{Pair}} = 9~ 000$, $p = 10 \%$, $N = 10$. The confidence level (type I error) used is $\gamma = 5 \%$.}
\label{bias_equal_p_optim_n}
\end{figure}

\begin{figure}[h]
\center
\includegraphics[width=11cm]{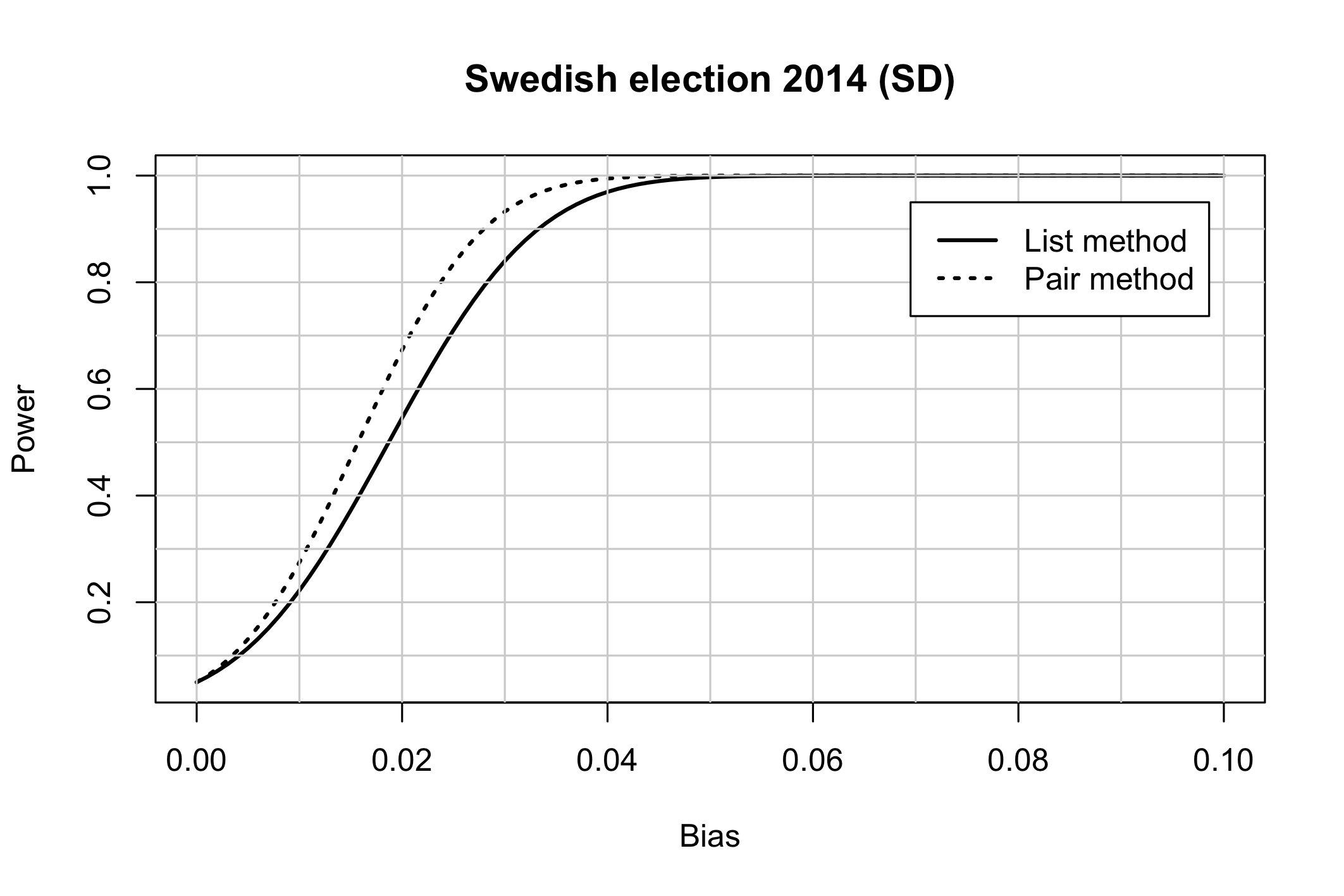}
\caption{Power calculation for bias-detection according to \eqref{eq : power-bias} where the total size of both surveys is $n = 15~ 000$ with $n_{\mathit{List}} = n_{\mathit{Pair}} = 13~ 500$, $p = 12.9 \%$ (SD) Swedish election 2014, $N = 10$. The confidence level (type I error) used is $\gamma = 5 \%$.}
\label{bias_sd_nbin1500}
\end{figure}

\begin{figure}[h]
\center
\includegraphics[width=11cm]{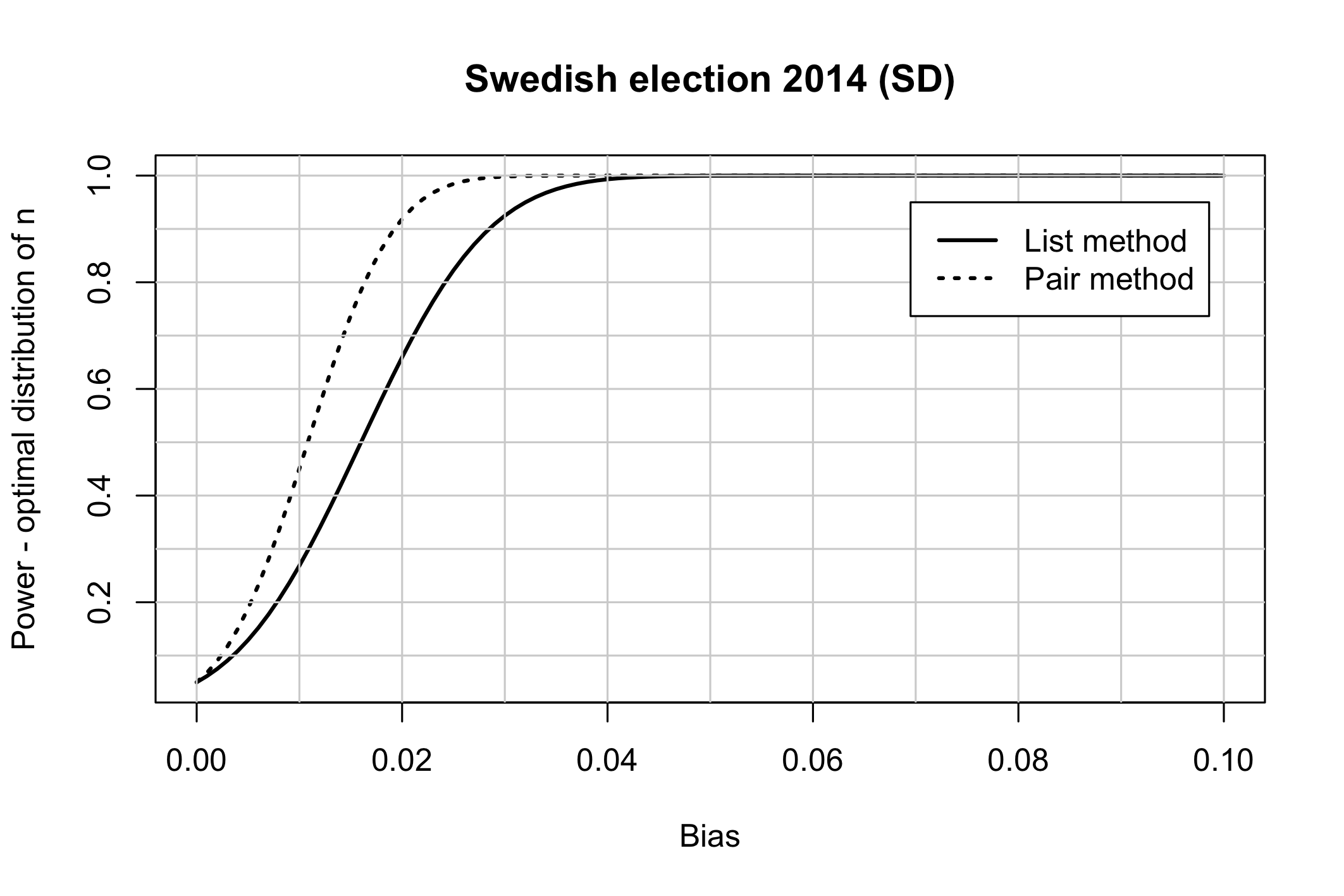}
\caption{Power calculation for bias-detection according to \eqref{eq : power-bias} where the total size of both surveys is $n = 15~ 000$, with optimised allocation between the standard binomial survey and the list and the pair method according to $n_{\mathit{List}} = 10~ 781$ and $n_{\mathit{Pair}} = 8~ 758$, $p = 12.9 \%$ (SD) Swedish election 2014, $N = 10$. The confidence level (type I error) used is $\gamma = 5 \%$.}
\label{bias_sd_optim_n}
\end{figure}

\end{document}